\documentclass[11pt]{article}
\usepackage[utf8]{inputenc}
\usepackage[letterpaper, margin=1in]{geometry}

\usepackage{amsmath,amsfonts,amssymb,amsthm,xspace}
\usepackage{graphicx}
\usepackage{paralist}
\usepackage{hyperref}
\usepackage{xcolor}
\usepackage{thm-restate}
\usepackage{wrapfig}
\usepackage{calrsfs}
\usepackage{lineno}

\newtheorem{theorem}{Theorem}
\newtheorem{lemma}[theorem]{Lemma}
\newtheorem{observation}[theorem]{Observation}
\newtheorem{definition}[theorem]{Definition}

\def\curveP{P}
\def\curveQ{Q}
\def\eps{\varepsilon}
\def\RR{{\mathbb R}}
\def\NN{{\mathbb N}}

\def\XX{{\mathbb X}}

\def\RSpace{{\pazocal{R}}}
\def\ASpace{{\widetilde{\pazocal{R}}}}

\def\Querytime{{T_{Q}}}
\def\Preproctime{{T_{P}}}

\def\Oraclespace{{S_{O}}}
\def\Curvesimpltime{{T_{S}}}
\def\approxcurvesimpl{{\sigma}}
\def\subcurveP{{\tilde{P}}}

\DeclareMathAlphabet{\pazocal}{OMS}{zplm}{m}{n}

\renewcommand{\emph}[1]{\textit{\textbf{#1}}}

\title{Subtrajectory Clustering:\\ Finding Set Covers for Set Systems of Subcurves}

\author{Hugo A. Akitaya\thanks{Department of Computer Science, University of Massachusetts Lowell, USA.} \and Frederik Br\"uning\thanks{Department of Computer Science, University of Bonn, Germany} \and Erin Chambers\thanks{Department of Computer Science, Saint Louis University, USA.} \and Anne Driemel\thanks{Hausdorff Center for Mathematics, University of Bonn, Germany.}}

\begin{document}
%\linenumbers
\maketitle
\pagestyle{plain}

\begin{abstract}
    We study subtrajectory clustering under the Fréchet distance. Given one or more trajectories, the task is to split the trajectories into several parts, such that the parts have a good clustering structure. We approach this problem via a new set cover formulation, which we think provides a natural formalization of the problem as it is studied in many applications. 
    Given a polygonal curve $P$ with $n$ vertices in fixed dimension,  integers  $k$, $\ell \geq 1$, and a real value $\Delta > 0$, the goal is to find $k$ center curves of complexity at most $\ell$ such that every point on $P$ is covered by a subtrajectory that has small Fréchet distance to one of the $k$ center curves ($\leq \Delta$). 
    In many application scenarios, one is interested in finding clusters of small complexity, which is controlled by the parameter $\ell$.
    Our main result is a bicriterial approximation algorithm: if there exists a solution for given parameters $k$, $\ell$, and $\Delta$, then our algorithm finds a set of $k'$ center curves of complexity at most $\ell$ with covering radius $\Delta'$ with $k' \in O( k \ell^2 \log (k \ell))$, and $\Delta'\leq 19 \Delta$. Moreover, within these approximation bounds, we can minimize $k$ while keeping the other parameters fixed.      
    If $\ell$ is a constant independent of $n$, then, the approximation factor for the number of clusters $k$ is $O(\log k)$ and the approximation factor for the radius $\Delta$ is constant. 
    In this case, the algorithm has expected running time in $ \tilde{O}\left( k m^2 + mn\right)$ and uses space in $O(n+m)$, where $m=\lceil\frac{L}{\Delta}\rceil$ and $L$ is the total arclength of the curve $P$.
\end{abstract}

\thispagestyle{empty}
\clearpage
\setcounter{page}{1}
\section{Introduction}

Trajectories appear in many different applications in the form of recorded sequences of positions of moving objects. A trajectory is usually modelled as a piecewise linear curve by interpolating between two consecutive location measurements. Standard examples include trajectories of migrating animals, sports players on the field, and vehicles in traffic~\cite{acmsurvey20, su2020survey}. Other examples include time series data from sensor measurements tracking the movement of a hand for gesture analysis~\cite{Qiao2017RealtimeHG}, or the focal point of attention during eye tracking \cite{ holmqvist2011eye, duchowski2002breadth}.
One particular question in trajectory analysis which has gotten much attention relates to clustering this type of data; typically, one wishes to extract patterns 
that summarize the data well. 
This necessitates a notion of similarity (or dissimilarity) to compare and evaluate simplified representations of curves.  The Fr\'echet distance is one such measure, which in addition to geometric closeness also takes the flow of the curve into account; see Section~\ref{sec:prelims} for the precise definition.

In this paper, we consider the problem of subtrajectory clustering. The main difference to standard trajectory clustering is that the input curves may be broken into subcurves by the clustering algorithm. Indeed, this approach is well motivated, as trajectory data is often collected over longer periods of time, and the start and ending of the trajectories often do not carry any particular meaning. In a sense, then, any particular trajectory might naturally break down into subtrajectories, for example when a car's route involves several independent stops as opposed to a single continuous trip.  A goal of subtrajectory clustering is  to find patterns within the trajectory data and to let an algorithm find the starting and ending points of these patterns by means of solving an optimization problem. 
There is much work on different variants of this subtrajectory clustering problem and many heuristics have been proposed, 
see also the surveys \cite{yuan2017review, BuchinW20, wang2021survey} and references therein. 
However, there does not seem to be a rigorous and commonly agreed upon definition of the underlying optimization problem. 

The purpose of this paper is to propose a class of problems that capture the nature of the subtrajectory clustering problem and provide algorithmic solutions with provable guarantees. Our work draws from ideas and techniques developed in works on the $(k,\ell)$-clustering variant for trajectories \cite{DriemelKS16,Buchin2019soda,nath2020k,Buchin2021soda}, where the complexity of centers is restricted by a parameter. We develop algorithmic techniques that build upon fundamental work on computing hitting sets of set systems for low VC-dimension.  In particular, we use the set cover framework algorithm by Br\"onniman and Goodrich~\cite{bronnimann1995almost}. This framework algorithm is related to the multiplicative weights update method~\cite{v008a006} and has been used in numerous applications. In computational geometry, it has been used for projective clustering~\cite{AGARWAL2003115} and the art gallery problem~\cite{gonzalesbanos2001}. It is related to Clarkson's algorithm for linear programming~\cite{clarkson1995vegas}, which predates it, see also the survey by Agarwal and Sharir~\cite{agarwal1998efficient}.
However, to the best of our knowledge, the framework has not been applied to subtrajectory clustering before. 
We remark that the algorithm by Brönniman and Goodrich~\cite{bronnimann1995almost} has been revisited and improved several times~\cite{agarwal2020near,ChanH20}, but our methods do not seem to profit from these improvements.

\subsection{Related work}\label{sec:relatedwork}

One of the earlier works on clustering subtrajectories is by Lee, Han and Whang~\cite{LeeHW07}. They were interested in computing a small set of line segments that describe the geometry of the input trajectories well. Their algorithm works in two phases: (i) a partition phase where they employ the minimum-description-length (MDL) principle and (ii) a grouping phase where they use a density based clustering algorithm similar to DBSCAN~\cite{dbscan96}. 

In general, it is not obvious how to combine the two phases---partitioning and grouping---into one optimization problem. Buchin et al.~\cite{buchin2011detecting} focus on the problem of finding one single cluster of subtrajectories that are similar to each other. More specifically, they define a subtrajectory cluster with parameters $s$, $\Delta$, and $\ell$, as a set of $s$ pairwise disjoint subtrajectories with pairwise Fr\'echet distance at most $\Delta$ and such that at least one of the subtrajectories has complexity at least $\ell$. They define three optimization problems that each optimize one of the three parameters while keeping the other two fixed. 
The decision problem where all three parameters are fixed is shown NP-complete via reduction from the MaxClique problem. 
They give $2$-approximation algorithms: \begin{inparaenum}[(i)]
   \item for finding the longest subtrajectory cluster (max $\ell$) and 
   \item for finding the subtrajectory cluster with the maximum number of subtrajectories (max $s$).
\end{inparaenum} 
In subsequent work, these algorithms have been used as building blocks in several heuristic algorithms for map construction~\cite{BuchinMapConstruction17,BuchinMapConstruction20}, where the task is to infer an underlying road map from a set of trajectories.

A natural way to define a global optimization criterion for subtrajectory clustering is by using the set cover problem: given a set of elements $X$ and a set of subsets $\RSpace \subseteq 2^X$, select a minimum number of sets from $\RSpace$, such that their union covers all of $X$.

Indeed, set cover formulations are used implicitly and explicitly in many algorithms for subtrajectory clustering. 
Buchin, Kilgus and K\"olzsch~\cite{buchinGroup20} study migration patterns of animals. They want to derive an augmented geometric graph (a so-called group diagram) that captures the common movement of a group of migrating animals. An input trajectory is represented in the group diagram if there exists a path in the graph that is similar to it under some predefined similarity measure, such as the Fr\'echet distance. Their algorithm constructs a set cover instance by extracting a linear number of subtrajectory clusters using the algorithm of~\cite{buchin2011detecting} (see the discussion above).  
Overall, the algorithm takes time in $O(k^3 N^3)$ given $k$ trajectories, each of at most $N$ vertices. However, they use a preprocessing phase that introduces additional vertices which may increase $N$ quadratically in the worst case leading to an overall running time of $O(k^3 N^6)$. The approximation factor for the number of clusters selected is $O(\log kN)$.

Agarwal et al.~\cite{agarwal2018} proposed a problem formulation based on facility location for subtrajectory clustering under the discrete Fr\'echet distance . They also consider a set cover problem as an intermediate step of their algorithm, but their formulation leads to a set system of exponential size. They present $O(\log n)$-approximation algorithms, where $n$ is the total number of vertices of the input curves. Their algorithm runs in $O(|B|n^3)$ if $B$ is a set of candidate center curves given with the input. They show how to generate a suitable set $B$ of size $O(n^2)$, and how to reduce the size to $O(n)$ at the expense of an additional $O(\log n)$-factor in the approximation quality.

\subsection{Organization}

In the remainder of this section, we give some preliminary definitions in Section~\ref{sec:prelims}, we define the problem statement in Sections~\ref{sec:def} and a modified problem statement in Section~\ref{sec:setup:setsystem}. We give an overview of our main results in Section~\ref{sec:overview} and discuss other problem variants in Section~\ref{sec:variants}. We then discuss our main techniques in Section~\ref{sec:setup}. In Sections~\ref{sec:main} and~\ref{sec:main:general} we discuss solutions to the modified problem. In Sections~\ref{sec:cont} we discuss our solution to the main problem stated in Section~\ref{sec:def}. In Sections~\ref{sec:vcdim} and~\ref{sec:nphard1} we discuss additional results.

\subsection{Preliminaries}\label{sec:prelims}

A sequence of $n$ points  $p_1,\dots,p_n \in \RR^d$ defines a \emph{polygonal curve} $\curveP$ by linearly interpolating consecutive points, that is, for each $i$, we obtain the \emph{edge} $\{ t p_i + (1-t) p_{i+1} | t \in [0,1]\}$. We may
think of $\curveP$ as resulting from the concatenation of the edges in the given order as a parametrized curve, that is, a function $\curveP: [0,1] \mapsto \RR^d$. Note that for any such parametrized curve there exist real values $s_1 \leq \dots \leq s_n$, such that $\curveP(s_i)=p_i$.  We call the ordered set of the $p_i$ the \emph{vertices} of $\curveP$ and we denote it with $V(\curveP)$. We call the number of vertices $n$ the \emph{complexity} of the curve.
For any two $[a,b] \subseteq [0,1]$ we denote with $\curveP[a,b]$ the \emph{subcurve} of $\curveP$ that starts at $\curveP(a)$ and ends at $\curveP(b)$. Let $\XX^d_{\ell} = (\RR^d)^{\ell}$, and think of
the elements of this set as the
set of all  polygonal curves of $\ell$ vertices in $\RR^d$.
For two parametrized curves $\curveP$ and $\curveQ$, we define their \emph{Fréchet distance} as
 \[ d_F(\curveP,\curveQ) = \inf_{\gamma:[0,1] \mapsto [0,1]} \sup_{t \in [0,1]}
\| \curveP(\gamma(t)) - \curveQ(t) \|, \] 

\noindent where $\gamma$ ranges over all strictly monotone increasing functions.
A curve $\curveQ \in \XX^d_{\ell}$ is called an \emph{$\ell$-simplification} of a curve $\curveP$ if its Fréchet distance is minimum among all curves in $\XX^d_{\ell}$. We denote with $\Curvesimpltime(n,\ell)$ the time needed to compute such an $\ell$-simplification for a polygonal curve of $n$ vertices.

Let $X$ be a set. We call a set $\RSpace$ where any $r \in \RSpace$ is of the form $r \subseteq X$ a \emph{set system} with \emph{ground set} $X$. 
Let $\RSpace$ be a set system with ground set $X$. A  \emph{set cover} of $\RSpace$ is a subset $S \subset \RSpace$ such that the ground set is equal to the union of the sets in $S$. The \emph{set cover problem} asks to find a set cover for a given $\RSpace$ using a minimum number of sets.
In addition to the set cover problem, we define the  \emph{hitting set problem}.
Let $\RSpace$ be a set system with ground set $X$. A \emph{hitting set} of $\RSpace$ is a subset $S \subseteq X$ such that every set of $\RSpace$ contains at least one element of $S$. The hitting set problem is to find a hitting set for a given $\RSpace$ of minimum size.
We denote with $\RSpace^{*}$ the set system \emph{dual} to $\RSpace$. The set system $\RSpace^{*}$ has ground set $\RSpace$ and  each set $r_x \in \RSpace^{*}$ is defined by an element $x \in X$ as  $r_x = \{ r \in \RSpace | x \in r \}$. The dual set system of $\RSpace^{*}$ is again $\RSpace$.
The hitting set problem for $\RSpace$ is equivalent to the set cover problem for the dual set system $\RSpace^{*}$.
We say a subset $A \subseteq X$ is \emph{shattered} by $\RSpace$ if for any $A' \subseteq A$ there exists an $r \in \RSpace$ such that $A'=r \cap A$. The \emph{VC-dimension} of $\RSpace$ is the maximal size of a set $A$ that is shattered by $\RSpace$. When stating asymptotic bounds, we may use the $\widetilde{O}(\cdot)$ notation hiding polylogarithmic factors to simplify the exposition.

\subsection{Problem definition}
\label{sec:def}

Let $\curveP:[0,1] \rightarrow \RR^d$ be a parametrized curve\footnote{We chose the setting of one input curve to keep the presentation of our algorithmic solutions as simple as possible. All of our algorithms can be easily extended to the setting of multiple input curves.} 
and let $\ell \in \mathbb{N}$ and $\Delta \in \RR$ be fixed parameters. 
Define the $\Delta$-\emph{coverage} of a set of center curves $C \subset \XX^d_{\ell}$ as follows:
\[ \Psi_{\Delta}(P,C) = \bigcup_{q \in C} ~ \bigcup_{0 \leq t \leq t' \leq 1} \{ s \in [t,t'] \mid  d_F(\curveP[t,t'], q)\leq \Delta   \}. \]

Note that this corresponds to the part of the curve $P$ that is covered by the set of all subtrajectories that are within Fréchet distance $\Delta$ to some  curve in $C$.
The problem we study in this paper is to find a set $C \subset \XX^d_{\ell}$ of minimum size such that the $\Delta$-coverage of $C$ covers all of $P$.
We define the \emph{radius} of the clustering induced by $C$ as the smallest real value $\Delta$ such that $\Psi_{\Delta}(P,C)= [0,1] $, and we denote the radius with $\psi(P,C)$.

\subsection{Set system formulation}\label{sec:setup:setsystem}

Our approach to the subtrajectory clustering problem (see Section~\ref{sec:def}) works via  set covers of suitable set systems. To this end, we will first define a \emph{discrete variant} of the problem. Assume that the curve $P$ is endowed with a set of $m$ real values $0 = t_1 < t_2 < \dots < t_m = 1$ which define a set of subcurves of the form $S_{ij} = P[t_i,t_j]$. We denote the set of values $t_i$ with $\pazocal{T}$ and we refer to the respective points on the curve $\curveP(t_i)$ for $1\leq i \leq m$ as \emph{breakpoints}. For a given curve $P$ with breakpoints we define the $\Delta$-coverage of a set of center curves $C \subset \XX^d_{\ell}$ with respect to these breakpoints as follows
\[ \Phi_{\Delta}(P,C) = \bigcup_{q \in C} ~ \bigcup_{1 \leq i \leq j \leq m} \{ s \in [t_i,t_j] \mid  d_F(\curveP[t_i,t_j], q)\leq \Delta   \} \]
Analogous to the problem definition in Section~\ref{sec:def} we define the radius of the clustering in the discrete case  as the smallest real value $\Delta$ such that $\Phi_{\Delta}(P,C)= [0,1] $ and we denote this radius with $\phi(P,C)$. 
Consider the set system $\RSpace$ with ground set $X = \{1,\dots,m-1\}$ where each set $r_{Q} \in
\RSpace$ is defined by a polygonal curve $Q \in \XX^d_{\ell}$ as follows
\begin{eqnarray} \label{eq:setsystem:exact}
r_{Q} = \left\{ z \in X \mid \exists  i \leq z < j  \text{ with }  d_F(Q,\curveP[t_i,t_j]) \leq \Delta \right\}
\end{eqnarray}

In the discrete case, the problem of finding a minimum-size set of center curves that cover $P$ now reduces to finding a minimum-size set cover for the set system $\RSpace$.
Stating the problem in terms of set systems allows us to draw from a rich background of algorithmic techniques for computing set covers (see Section~\ref{sec:setup}).

We first discuss how solutions to the discrete problem help solving the initial problem defined in Section~\ref{sec:def}. We can choose breakpoints for the input curve $P$, such that the distance between two consecutive breakpoints is at most $\epsilon \Delta$ for any fixed $\epsilon>0$. This is always possible with $m=\lceil \frac{L}{\epsilon \Delta}\rceil$ breakpoints, where $L$ is the arclength of $P$. 
The resulting instance of the discrete problem variant approximates the continuous version of the problem in the following way.

\begin{lemma}\label{lem:main:approx}
Assume there exists a set $C^{*}\subset \XX^d_{\ell}$ of size $k$, such that $\psi(P,C^{*})\leq \Delta$. Then we have $\phi(P,C^{*})\leq (1+\epsilon)\Delta$. Additionally for each $C\subset \XX^d_{\ell}$ with $\phi(P,C)\leq (1+\epsilon)\Delta$ we have $\psi(P,C^{*})\leq (1+\epsilon)\Delta$.
\end{lemma}
\begin{proof}
We show that for any set of center curves  $C\subset \XX^d_{\ell}$  we have $\psi(P,C)\leq \phi(P,C)\leq (1+\epsilon)\psi(P,C)$.
Indeed, if $C$ covers the curve $P$ in the discrete setting, then it also covers the curve $P$ in the continuous setting. Therefore, $\psi(P,C)\leq \phi(P,C)$.
For showing the other inequality we observe that the distance between two consecutive breakpoints is at most $\epsilon\Delta$. Therefore, for any interval $[s,t] \subset [0,1]$ we can choose breakpoints $t_i\leq s$ and $t_j \geq t$ such that $d_F(P[s,t],P[t_i,t_j])\leq \eps \Delta$.
The claim now follows from the triangle inequality.
\end{proof}

\subsection{Main results}\label{sec:overview}

We study the problem of subtrajectory clustering in the concrete form as defined in Section~\ref{sec:def}. We think that this problem formulation provides a natural formalization of the problem as it is studied in many applications (see also the discussion in Section~\ref{sec:variants}). We develop bicriterial approximation algorithms for this problem, where the approximation is with respect to the following two criteria
\begin{inparaenum}[(i)]
    \item the number of clusters $k$, and
    \item the radius of the clustering $\Delta$.
\end{inparaenum}

In Sections~\ref{sec:main} and \ref{sec:main:general} we describe our approach for the discrete variant of the subtrajectory clustering problem defined in Section~\ref{sec:setup:setsystem}, before we turn to the main problem in Section~\ref{sec:cont}. 
We first discuss the special case where cluster centers are restricted to be directed line segments (the case $\ell=2$, Theorem~\ref{thm:main:lines}). The main idea is to define a suitable set system that preserves optimal solutions up to approximation and at the same time allows for efficient set system oracles. A set system oracle is a data structure that answers queries with a set $r$ and an element of the ground set $x$ and returns whether $x \in r$. 
We solve this by defining a linear number of ``proxy'' curves which are simplifications of subcurves that are locally maximal. The proxy curves allow to solve a set system query by computing a partial Fr\'echet distance with some additional conditions.  
%Note that the above theorem postulates a bi-criterial approximation algorithm.
In the more general case, where cluster centers can be curves of complexity $\ell > 2$, we use the bi-criterial simplification algorithm of Agarwal et al.~\cite{Agar05} to define suitable proxy curves. This is described in Section~\ref{sec:main:general} and the result is stated in Theorem~\ref{thm:main}.

Finally, in Section~\ref{sec:cont}, we present our solution to the main problem of subtrajectory clustering, where subtrajectories can start and end at any two points along the curve (see Section~\ref{sec:def}). We use the techniques developed in Section~\ref{sec:main:general}, but we obtain better approximation factors and running times, compared to a naive application of Lemma~\ref{lem:main:approx}. The improved running time results from the fact that we do not need to keep track of breakpoints explicitly in the set system oracle. 
Crucial to obtaining better approximation factors is the analysis of the VC-dimension of the dual set system. 
We obtain the following theorem. 

\begin{restatable}[Main Theorem]{theorem}{mainTheoremcont}\label{thm:main:cot}
Let $\curveP: [0,1] \rightarrow \RR^d$ be a polygonal curve of complexity $n$, let $\ell \in \NN$ and $\Delta,\eps > 0$ be parameters. Assume there exists a set $C^{*}\subset \XX^d_{\ell}$ of size $k$, such that $\psi(P,C^{*})\leq \Delta$. Let $m=\lceil\frac{L}{\epsilon\Delta}\rceil$ and $\delta= O(d^2\ell^2\log(d\ell)))$, there exists an algorithm that computes a set $C \subset \XX^d_{\ell}$ of size $O( k\delta \log (\delta k))$, such that $\psi(P,C)\leq (18+\epsilon)\Delta$. The algorithm has expected running time in
$ \widetilde{O}\left(
k m^2 + mn\right
)$
and uses space in $O(n+m)$, where we assume that $\ell$ and $d$ are constants independent of $n$.
\end{restatable}

In particular, in the above theorem, when the complexity of center curves $\ell$ and the ambient dimension $d$ are constants, the VC-dimension $\delta$ is constant, and the approximation factor for the size of the set cover is $O(\log k)$.
For a comparison, using Theorem~\ref{thm:main} and Lemma~\ref{lem:main:approx} directly would result in an approximation factor of $O(\log \frac{L}{\Delta}\log^2 \frac{L}{\Delta})$ which could be large even if $\ell$ and $d$ are small.

The improved approximation factors that are obtained in the continuous case in Theorem~\ref{thm:main:cot} raise the question if the approximation factor could be improved in the discrete case. Unfortunately, this does not seem to be the case. In Section~\ref{sec:vcdim} we study lower bounds to the VC-dimension for two natural problem variants. We study the dual set system (i) in the discrete case and (ii) the set system directly corresponding to our main clustering problem. For (i) we show a lower bound of $\Omega(\log m)$ directly corresponding to the upper bound, see Theorem~\ref{thm:vcdim:discrete}. For (ii) we show that---surprisingly---it inherently depends on the number of vertices of the input curve $n$, even when cluster centers are restricted to be line segments, see Theorem~\ref{thm:vcdim:cont} for the exact statement. Thus, ultimately, our modified set system with proxy curves not only makes the algorithm faster, but also has the benefit of a significantly lower VC-dimension, compared to the exact set system inherent to the problem.

%\paragraph{Hardness}
Finally, we also investigate the question of hardness for the discrete problem defined in Section~\ref{sec:setup:setsystem}. If the complexity of center curves $\ell$ can be large, then NP-hardness follows from the hardness of the shortest common superstring problem, see also the result by Buchin et al.~\cite{Buchin2019soda} on $(k,\ell)$-center clustering under the Fr\'echet distance. In particular, in this case the problem is also hard to approximate. In Section~\ref{sec:nphard1} we show that even if we require cluster centers to be points by setting $\ell=1$, the problem remains NP-hard, via a reduction from \textsc{Planar-Monotone-3SAT}.

\subsection{Other problem variants and future directions}\label{sec:variants}

We think that our definition of subtrajectory clustering given in  Section~\ref{sec:def} captures a fundamental problem that arises in many applications. However, one may argue that there are many other variants of subtrajectory clustering that arise from application-specific considerations, see also the discussion of related work in Section~\ref{sec:relatedwork}. We expect that our general approach and our problem definition can be applied to many of these variants. 
For example, all of our algorithms can be easily extended to the setting of multiple input curves.
We mention some other variants that we find interesting.

\begin{enumerate}[(1)]
    \item Outputting a graph: The output of our algorithm is a set of center curves. In some applications, such as map construction, we may prefer the output to be a geometric graph. This can be easily obtained by connecting the center curves to form a geometric graph using additional edges where the input trajectory moves from one cluster to the next. How to do this optimally would a subject for future research.
    \item Covering with gaps: One might be interested in a problem variant where not the entire curve needs to be covered, but only a certain fraction of the curve. It would be interesting to analyze our techniques in this setting.
    \item Input curves: In this paper, we assume that our input curves are given in the form of polygonal curves. However, it is conceivable that our general approach to the discrete problem still works if the input is given in the form of piecewise polynomial curves with breakpoints; again, we leave this to future work.
    \item Other distance measures: Similarly, we think that the general approach to the discrete problem, where breakpoints are given with the input, is still applicable, if the Fr\'echet distance is replaced by some other distance measure that satisfies the triangle inequality.  
\end{enumerate}

It is tempting to relax the restriction on the complexity of the center curves in our problem definition. However, without any other regularization of the optimization problem, this would lead to the trivial solution of the curve $P$ being an optimal center curve. 
In any case, we think that some form of controlled regularization is necessary in the problem definition.

\section{Setup of techniques}\label{sec:setup}

In this section 
we introduce the main ideas and concepts that we use in our algorithms.

We start in Section~\ref{sec:setup:approx} with a simple algorithm that illustrates our general approach in a nutshell: we derive an auxiliary set system that has a simpler structure and smaller size compared to the set system of Section~\ref{sec:setup:setsystem}, while preserving optimal solutions up to approximation. A preliminary result that follows by applying the greedy set cover algorithm is stated in Theorem~\ref{thm:result2}. 
Then, in Section~\ref{sec:setup:algorithm} we recapitulate the algorithmic framework by Brönniman and Goodrich~\cite{bronnimann1995almost} which we use in our main algorithm. In order to obtain efficient algorithms from this framework, we need to adapt the framework to our specific needs. 
%The technical details of this adaptation are diverted to Section~\ref{sec:apendix:framework}.
The approximation quality of the resulting algorithm strongly depends on the VC-dimension of the dual set system. Therefore, we aim for auxiliary set systems with constant VC-dimension. Alas, this is not always possible when breakpoints are given with the input. We discuss this in Section~\ref{sec:setup:vcdim}.

\subsection{A set system for approximation}\label{sec:simplification}\label{sec:setup:approx}

In this section we discuss a simple algorithmic solution to the discrete variant of the problem we study. We emphasize that the approach works for any choice of breakpoints and is thus interesting in its own right. The algorithm yields a bicriteria approximation in the radius $\Delta$ and the number of clusters $k$. Although this algorithm is suboptimal, we include it here as an illustration of our general approach to the subtrajectory clustering problem: modify the set system in a way that preserves the initial structure up to approximation but allows for more efficient algorithms for the clustering problem. 

Let $\pazocal{S} = \{(i,j) \in \NN^2 \mid 1 \leq i < j \leq m \}$. For any $(i,j) \in \pazocal S$  let  $\mu_{\ell}(P[t_i,t_j])$ denote the $\ell$-simplification of the corresponding subcurve of $P$.
Consider a set system $\ASpace_0$ defined on the ground set $X=\{1,\dots,m-1\}$, where each set $r_{i,j} \in \ASpace_0$ is defined by a tuple $(i,j) \in \pazocal{S}$ and is of the form
\[ r_{i,j} = \{ z \in X \mid \exists  i' \leq z < j'  \text{ with } d_F(P[t_{i'},t_{j'}],\mu_{\ell}(P[t_i,t_j])) \leq 3\Delta \}\]

We will see (Lemma~\ref{lem:hitting:simplified}, below), that $\ASpace_0$ approximates the structure of $\RSpace$ as defined in (\ref{eq:setsystem:exact}) to the extent that a set cover for $\ASpace_0$ corresponds to an approximate solution for our clustering problem. 
The well-known greedy set cover algorithm, which incrementally builds a set cover by taking the set with the largest number of still uncovered elements in each step, yields an $O(\log m)$ approximation for a ground set of size $m$ \cite{chvatal979greedy}.  Applying this algorithm to the set system $\ASpace_0$, we obtain a set $C$ consisting of $\ell$-simplifications $\mu_{\ell}(P[t_i,t_j]))$ for each $r_{i,j}$, such that  $\phi(P,C)\leq 3\Delta$.

\paragraph{Building the incidence matrix} To this end, we compute the binary incidence matrix $M$ of the set system $\ASpace_0$ explicitly in $O(m^3(n\ell+m) + m^2 \Curvesimpltime(n,\ell))$ time, as follows.
Initially we set all entries of the matrix to $0$. 
In the first step we compute the $O(m^2)$ simplifications $\mu_{\ell}(P[t_i,t_j])$ of all subcurves between two breakpoints. For each simplification $\mu$, we compute the $\Delta$-free space with the curve $\curveP$, which is defined as the level set 
\[ FD_{\delta}(\curveP, \mu) = \left\{ (x,y) \in [0,1]^2 \mid {  \|{\curveP(x)} - {\mu(y)}\| \leq \delta} \right\}.  \] 
Computing the associated diagram can be done in $O(n \ell)$ time and space~\cite{AltG95}.
Note that the simplification $\mu$ corresponds to the vertical axis of the $\Delta$-free space diagram and $\curveP$ corresponds to the horizontal axis. 
Now, for each breakpoint $t_{i'}$ we compute the maximal breakpoint $t_{j'}$ that is reachable by a monotone path from the bottom of the diagram at $(t_{i'},0)$ to the top of the diagram at $(t_{j'},1)$. This can be done in $O(n\ell)$ time using standard techniques~\cite{AltG95}. For all $i' \leq q < j'$, we set the entry corresponding to $q$ and $\mu$ to $1$. This takes $O(m)$ time. We do this for all simplifications. After that, each entry of $M$ is $1$ if the corresponding element is contained in the corresponding set and $0$ otherwise.

\paragraph{Applying greedy set cover} We initially scan the incidence matrix to  compute the number of uncovered elements $n_{i,j}$ for every range $r_{i,j} \in \ASpace_0$. After this, we can compute the set with the highest number of uncovered elements in $O(m^2)$ time. Then, we can update all $n_{i,j}$ on the fly every time we select a new set for the set cover. To do so, we scan for each newly covered element all the $m^2$ entries of the incidence matrix corresponding to this element and reduce $n_{i,j}$ by $1$ if the entry corresponding to $r_{i,j}$ is equal to $1$. Since each of the $m$ elements gets covered for the first time only once, this can be done in a total time of $O(m^3)$. 

\begin{lemma}\label{lem:hitting:simplified}
For any $r_Q \in \RSpace$, there is a $r_{i,j} \in \ASpace_0$ such that $r_Q \subseteq r_{i,j}$.
\end{lemma}

\begin{proof}
We can rewrite the definition of $r_Q$ as follows.
Let $Y$ be the set of tuples $(i,j) \in \NN^2$ 
with $1 \leq i < j \leq m$ and  $d_F(Q,\curveP[t_i,t_j]) \leq \Delta$. We have that
$r_{Q} =  \bigcup_{(i,j) \in Y} [i,j) \cap \NN $. Let $(i,j),(i',j') \in Y$. Using the triangle inequality, we can upper bound $d_F(P[t_{i'},t_{j'}],\mu_{\ell}(P[t_i,t_j]))$ by
\[ d_F(P[t_{i'},t_{j'}],Q) + d_F(Q,P[t_{i},t_{j}]) +
d_F(P[t_{i},t_{j}],\mu_{\ell}(P[t_i,t_j]))\leq 3 \Delta. \]

By the definition of $r_{i,j}$, we have $[i',j') \cap \NN \subseteq r_{i,j}$ and therefore $r_Q \subseteq r_{i,j}$. In other words, we can choose any maximal set of covered intervals within $r_Q$ and use the simplification of the corresponding subcurve of $\curveP$ to cover all parts of $P$ that are covered by $Q$. 
\end{proof}

\begin{theorem}\label{thm:result2}
Given a polygonal curve $\curveP: [0,1] \rightarrow \RR^d$ with breakpoints $0 \leq t_1,\dots,t_m\leq 1$. Assume there exists a set of curves $C^*  \subset \XX^d_{\ell}$ of size $k$, such that $\phi(P,C^*)\leq \Delta$.   There exists an algorithm that computes a set $C \subset \XX^d_{\ell}$ of size $O(k \log m)$ and has running time in 
$O(m^3 n \ell +m^4 + m^2 \Curvesimpltime(n,\ell)))$ 
such that $\phi(P,C)\leq 3\Delta$,
where $\Curvesimpltime(n,\ell)$ denotes the the running time for computing an $\ell$-simplification of a polygonal curve of $n$ vertices.
\end{theorem}

\begin{proof}
The algorithm builds the incidence matrix of the set system and applies greedy set cover, as described above. The bound of the running time is immediate. It remains to argue correctness. 
The existence of a set of curves $C^*$ of size $k$ with $\phi(P,C^*)\leq \Delta$ implies that there exists a set cover of $\RSpace$ of size $k$. Lemma~\ref{lem:hitting:simplified} implies that for any set cover of $\RSpace$, there exists a set cover of $\ASpace_0$ of the same size.
Thus, the $O(\log m)$-approximate set cover $S$ computed by the algorithm for $\ASpace_0$ has size at most $O(k \log m)$. Let
\[C = \{\mu_{\ell}(P[t_i,t_j]) \mid r_{i,j} \in S \}. \]

Since $S$ is a set cover for $\ASpace_0$, and by the definition of $r_{i,j}$, we have $\phi(P,C)\leq 3\Delta$.
\end{proof}

\subsection{The framework for the set cover algorithm}\label{sec:setup:algorithm}

For obtaining our main results we use the set cover framework algorithm by Br\"onnimann and Goodrich described in \cite{bronnimann1995almost} for set systems of low VC-dimension.  The idea of the framework algorithm is best explained by taking the point of view that it computes a hitting set of the dual set system. We need the following definition of an $\eps$-net.

\begin{definition}
Let $\RSpace$ be a set system with finite ground set $X$ and with an additive weight function $w$ on $X$. An \emph{$\eps$-net} is a subset $S \subset X$, such that every set of $\RSpace$ of weight at least $\eps\cdot w(X)$ contains at least one element of $S$. 
\end{definition}

Note that, if $w(x)=1$ for each $x\in X$, then an $\eps$-net is a hitting set for the ``heavy'' sets of $\RSpace$ that contain at least an $\eps$-fraction of the ground set.

\begin{definition}[\cite{bronnimann1995almost}] The framework algorithm needs the following subroutines. 
A \emph{net finder} of size $s$ for a set system $(X,\RSpace)$ is an algorithm $A$ that, given $\rho \in \mathbb{R}$ and a weight function $w$ on $X$, returns an $(1/\rho)$-net of size $s(\rho)$ for $(X,\RSpace)$ with weight $w$. Also, a \emph{verifier} is an algorithm $B$ that, given a subset $H \subseteq X$, either states (correctly) that $H$ is a hitting set, or returns a nonempty set $r$ of $\RSpace$  such that $r \cap H = \emptyset$. 
\end{definition}

\paragraph{Framework}  Given these two subroutines and a finite set system $(X,\RSpace)$, the algorithm proceeds as follows. In each iteration, the algorithm calls the net finder to compute an $\eps$-net $S$ of $\RSpace$ (for a specific value of $\eps$). Then, the algorithm calls the verifier to test if $S$ is also a hitting set for $\RSpace$. If yes, we return $S$. If no, then the verifier returns a witness set $r$ that does not contain any element of $S$. We increase the weight of each element of $r$ by a factor of $2$. Then, we repeat until we find a hitting set. 

%In Section~\ref{sec:apendix:framework}, 
We describe an easy adaptation of this framework algorithm that suits our needs. Our adaptation uses the following definition of a set system oracle. The resulting theorem is stated below. 
%The proof is an adaptation of the proof by Brönniman and Goodrich~\cite{bronnimann1995almost}.
%and can be found in Section~\ref{sec:apendix:framework}.

\begin{definition}[Set system oracle]
For a given set system $\RSpace$ with ground set $X$ a \emph{set system oracle} is a data structure $\pazocal{D}$ that can be queried with any $r \in \RSpace$ and $z \in X$ and answers whether $z \in r$. We denote with $\Preproctime(\pazocal{D})$ the preprocessing time to build the data structure $\pazocal{D}$ for the oracle and with $\Querytime(\pazocal{D})$ the time needed to answer the query. We denote with $\Oraclespace(\pazocal{D})$  the space required by the data structure.
\end{definition}

\begin{restatable}[Folklore]{theorem}{frameworkhittingset}
\label{thm:frameworkhittingset} \label{thm:framework}
For a given finite set system $(X,\RSpace)$ with finite VC-dimension $\delta$, assume there exists a hitting set of size $k$. Then, there exists an algorithm that computes a hitting set of size $k' \in O(\delta k \log \delta k)$ with expected running time in 
$ O\big( \left( k'|\RSpace| +|X| \right) k \log\left(|X| \right)  \Querytime(\pazocal{D})   +\Preproctime(\pazocal{D})\big)$
 and using space in $O\big(|X|+\Oraclespace(\pazocal{D})\big)$. 
\end{restatable}

\label{sec:apendix:framework}
\label{sec:simpleimpl}
\label{sec:genfram} 

In the remainder of this section we show how to prove Theorem~\ref{thm:frameworkhittingset} using an argument by Brönniman and Goodrich~\cite{bronnimann1995almost}. 

Let $\RSpace$ be a set system with finite ground set $X$ and finite VC-dimension $\delta$. An effective way to implement the net-finder is via a random sample from the ground set, as guaranteed by the $\eps$-net theorem~\cite{haussler1987eps} by Haussler and Welzl. 

\begin{theorem}[\cite{haussler1987eps}]
For any $(X,\RSpace)$ of finite VC-dimension $\delta$, finite $A \subseteq X$ and $0<\epsilon$, $\alpha<1$, if $N$ is a subset of $A$ obtained by at least
\[   \max\left(\frac{4}{3} \log\left(\frac{2}{\alpha}\right),\frac{8\delta}{\epsilon}\log\left(\frac{8\delta}{\epsilon}\right)\right) \]
random independent draws, then $N$ is an $\epsilon$-net of $A$ for $\RSpace$ with probability at least $1-\alpha$.
\end{theorem}

Thus, the net-finder can be implemented to run in $O(|X|)$ time and $O(|X|)$ space, by taking a sample from $X$ where the weights correspond to probabilities. We call this the \emph{probabilistic net-finder}.
While verifying that a set is an $\eps$-net could be costly in our setting, we can observe that this is actually not necessary. Indeed, we can modify the behaviour of the verifier as follows. 

\begin{definition}[Extended verifier]
Given a set $S \subseteq X$, the extended verifier returns one of the following:
\begin{compactenum}[(i)]
\item $S$ is a hitting set. 
\item A witness set $r$ with $r \cap S = \emptyset$, and $w(r) \leq \eps$. 
\item A witness set $r$ with $r \cap S = \emptyset$, and $w(r) > \eps$. 
\end{compactenum}
\end{definition}

To implement the extended verifier  we assume that we have a set system oracle $\pazocal{D}$ for $(X,\RSpace)$.
%To implement the extended verifier  we assume that we have an oracle that returns for any given $r \in \RSpace$ and $z \in X$  if $z \in r$. 
%Let $\Preproctime(\pazocal{D})$ denote the preprocessing time to build a data structure $\pazocal{D}$ for the oracle and let $\Querytime(\pazocal{D})$ denote the time needed to answer the query $z \in r$. Let further $\Oraclespace(\pazocal{D})$ be the space required for the data structure of the oracle.
After reprocessing the oracle, the extended verifier can be implemented to run in \[ O(|S|\cdot |\RSpace| \cdot \Querytime(\pazocal{D}) + |X|\cdot \Querytime(\pazocal{D}))\] time by using $|\RSpace|$ linear scans over $S$, one for each set in $\RSpace$. We determine for every set $r \in \RSpace$ whether it is hit by an element of $S$, by calling the set system oracle on $r$ and the corresponding set and elements in $S$. If we find a set that is not hit by any of the elements in $S$, we compute its weight explicitly by using $|X|$ calls to $\Querytime(\pazocal{D})$ and return the appropriate answer (ii) or (iii). In case (ii), we return the witness set that we have just computed explicitly, that is, we return all elements of this set, in order for the reweighting to be applied. If we do not find such a set, then $S$ is a hitting set and we return (i).

\paragraph{Algorithm.} Using the above implementation of the probabilistic net-finder and the extended verifier, the algorithm for computing a hitting set now proceed as follows. In each iteration we use the probabilistic net-finder to sample a candidate set $S\subseteq X$. The sample size is chosen large enough that $S$ is an $\eps$-net with probability greater $\frac{1}{2}$ (for a specific value of $\eps$). Given $S$, we apply extended verifier. If the verifier returns that $S$ is a hitting set (case (i)) then the algorithm terminates with $S$ as a result. If the verifier returns a witness set $r$ with $r \cap S = \emptyset$, and $w(r) \leq \eps$ (case (ii)) then we increase the weight of each element of $r$. The algorithm keeps track of the weight of each element and the weight of the whole ground set . If the verifier returns a witness set $r$ with $r \cap S = \emptyset$, and $w(r) > \eps$ (case (ii)) then $S$ is not an $\eps$-net and we do not change anything. We repeat these steps until we find a hitting set.

%\frameworkhittingset*

\begin{proof}[Proof of Theorem~\ref{thm:frameworkhittingset}]
 We first build a data structure $\pazocal{D}$ for the oracle in $\Preproctime(\pazocal{D})$ time. Then we use the algorithm described above
  with $\epsilon=\frac{1}{2k}$.
In each iteration of the algorithm the computed random sample of size $O(k \delta \log(\delta k))$  is an $\epsilon$-net with probability greater $\frac{1}{2}$. Therefore the expected number of iterations until we find an $\epsilon$-net is at most $2$. 

If we find an $\epsilon$-net in some iteration of the algorithm then we are in case (i) or (ii). As soon as we are in case (i) the algorithm terminates and outputs a hitting set of size $O(\delta k \log (\delta k))$. 

Let $H$ be a hitting set of $\RSpace$ with $|H|=k$. The number of times we can be in case (ii) before being in case (i) is bounded by $4k \log(\frac{|X|}{k})$. Indeed, after this number of reweighting steps, the weight of $H$ would be bigger than the weight of the ground set. The calculation of this bound has already been done in  \cite{bronnimann1995almost}. We include it here for the sake of completeness.

Let $r$ be the set returned by the verifier in one iteration of being in case (ii). Since $H$ is a hitting set, we have $H\cap r  \neq \emptyset$. Let $w$ be our weight function and let $z_h$ be the number of times the weight of $H$ has been doubled after $i$ iterations in case (ii). Then we have after $i$ iterations in case (ii) that
\[ w(H) = \sum_{h \in H} 2^{z_h}, \text{ where } \sum_{h \in H} z_h \geq i.\]

By the convexity of the exponential function, we get $w(H)\geq k 2^{\frac{i}{k} }$. Since $\epsilon=\frac{1}{2k}$, we also have for the ground set $Z$ that
\[ w(X)\leq |X|\left(1+\frac{1}{2k}\right)^i \leq |X|e^{\frac{i}{2k}}.\]

Because $H$ is a subset of $X$ and therefore $w(H)\leq w(X)$, we get in total 
\[ k 2^{\frac{i}{k}} \leq |X|e^{\frac{i}{2k}} \leq |X|2^{\frac{3i}{4k} }.\]

It directly follows that $i\leq 4k \log(\frac{|X|}{k})$.
 Combining this result with the expected number of iterations until we find an $\epsilon$-net, we conclude that the expected number of iterations before the algorithm terminates is smaller than $8 k \log(\frac{|X|}{k})$.

In each iteration, the algorithm computes a random sample in $O(|X|)$ time and applies the extended verifier in $O(|\RSpace| k \delta \log(k)\Querytime(\pazocal{D})+|X|\Querytime(\pazocal{D}))$ time. If a reweighting needs to be applied (case (ii)) this can be done in $O(|X|)$ time. So each iteration of the algorithm has a running time of $O(|\RSpace| \delta k \log(\delta k)\Querytime(\pazocal{D})+|X|\Querytime(\pazocal{D}))$.
In total we get an expected running time of 
\[ O(k \log(\frac{|X|}{k})\Querytime(\pazocal{D})(\delta k\log(\delta k)|\RSpace| +|X|) +\Preproctime(\pazocal{D})).\]
The theorem follows by the observation that both the net-finder and the verifier need $O(|X|)$ space.
\end{proof}

\subsection{Bounding the VC-dimension}\label{sec:setup:vcdim}

In order to use Theorem~\ref{thm:framework}  of Section~\ref{sec:setup:algorithm}, we need to bound the VC-dimension of the dual set system. In our case, this will be a set system that has similar structure as a set system of metric balls under the Fr\'echet distance  studied by Driemel et al.~\cite{driemel2019vc}. In a nutshell, they showed a bound of $O(d^2 s^2 \log(ds))$ for polygonal curves in $\RR^d$ of complexity at most $s$. Using this result directly would not gain us any useful bounds, as the subcurves $P[t_i,t_j]$ in the definition of the set system may have linear complexity in $n$---even for the simpler variant of Section~\ref{sec:setup:approx}. In fact, it turns out that the VC-dimension of the dual set system for the main problem defined in Section~\ref{sec:def} does indeed inherently depend on $n$, as we show in Theorem~\ref{thm:vcdim:cont} in Section~\ref{sec:vcdim:cont}. 

In Section~\ref{sec:cont} we instead define an auxiliary set system that preserves solutions up to approximation and---more importantly---which has low VC-dimension in the dual. We show this by using the approach of Driemel et al.~\cite{driemel2019vc}. They derive a set of geometric predicates which specify sufficient information for evaluating whether the Fr\'echet distance is below a certain threshold. Based on this, they define a composite set system that uses the geometric predicates as building blocks. The VC-dimension can then be bounded using standard composition arguments in combination with a theorem  by Anthony and Bartlett \cite{antbart99}.
Our analysis of the VC-dimension is given in Section~\ref{sec:cont:vcdim} and relies on the same set of geometric predicates. We relate these predicates to the distance evaluation of a certain type of partial Fr\'echet distance with specific conditions that occur in our set system with proxy curves. The result is stated in Theorem~\ref{thm:vc:cont} and implies that the VC-dimension is constant, if the complexity of the center curves $\ell$ and the ambient dimension $d$ is constant.

One may ask if a similar bound can be proven in the case where breakpoints are given with the input. Trivially, the size of the set system already gives a bound of $O(\log m)$, however this depends on the number of breakpoints $m$ and can be large even if $\ell$ is small.
We study this problem in  Section~\ref{sec:vcdim:discrete}. For the set system defined in Section~\ref{sec:setup:setsystem} we show a lower bound of $\Omega(\log m)$ even in the case that $d=1$ and $\ell=2$ (see Theorem~\ref{thm:vcdim:discrete}). Technically, this does not rule out the existence of an auxiliary set system with low VC-dimension in the dual. However, it is not clear what such a set system would look like as Theorem~\ref{thm:vcdim:discrete} makes only few assumptions on the set system. Thus, perhaps surprisingly, the discretization with breakpoints which was supposed to simplify the problem, actually makes it more difficult. Therefore, our approximation guarantee in the continuous case is better  to what we can currently achieve in the discrete case, when breakpoints are given with the input.

\section{Warm-up --- Clustering with line segments}
\label{sec:alg3}\label{sec:main}

In this section, we  show how to apply Theorem~\ref{thm:framework} to the discrete problem where we are given a curve $P$ with breakpoints. We assume in this section that cluster centers are restricted to be line segments (the case $\ell=2$). The general case ($\ell \geq 2$) is discussed in Section~\ref{sec:main:general}. In contrast to the solution described in Section~\ref{sec:setup:approx}, our algorithm finds an approximate set cover without computing the set system explicitly leading to better running times.

\subsection{The set system}\label{subsec:setsysline}

We start by defining the set system $\ASpace_2$ with ground set $Z=\{ 1, \dots, m-1 \}$.
Denote $\tau_{i,j} = \overline{{P(t_i) P(t_j)} }$.
For a subsequence $S=s_1,\dots,s_r$ of $1,\dots,m$, denote 
\[\pi(S) = \tau_{s_1,s_2} \oplus \tau_{s_2,s_3} \oplus \dots \oplus \tau_{s_{r-1},s_{r}}.\]

A tuple $(i,j)$ with $1\leq i \leq j \leq m$ defines a set $r_{i,j} \in \ASpace_2$ as follows
\[ r_{i,j} = \{ z \in Z \mid  \exists x \in [x_z, z], y \in [z+1,y_z] \text{ with } d_F(\pi({x,z,z+1,y}), \tau_{i,j}) )  \leq 2\Delta \},  \]
where $x_z \leq z < y_z $ are indices which we obtain as follows.
We scan breakpoints starting from $z$ in the backwards order along the curve and to test for each breakpoint $x$, whether 
\begin{eqnarray}\label{eq:scan}
d_F(\tau_{x,z},  P[t_x,t_z]) \leq 4\Delta.
\end{eqnarray} 
If $x$ satisfies~(\ref{eq:scan}), then we decrement $x$ and continue the scan. If $x=0$ or if $x$ does not satisfy~(\ref{eq:scan}), then we set $x_z = x+1$ and stop the scan. 
To set $y_z$ we use a similar approach: We scan forwards from $z+1$ along the curve and test for each breakpoint $y$ the same property with $\tau_{z+1,y}$ and $P[t_{z+1},t_{y}]$. If $y$ satisfies the property, we increment $y$ and continue the scan. If $y=m+1$ or if $y$ does not satisfy the property we set $y_z=y-1$ and stop the scan. Figure~\ref{fig:hittingwithlines} shows an example of $z,x_z$ and $y_z$.

\begin{figure}[t]
    \centering
    \includegraphics[width=0.6\textwidth]{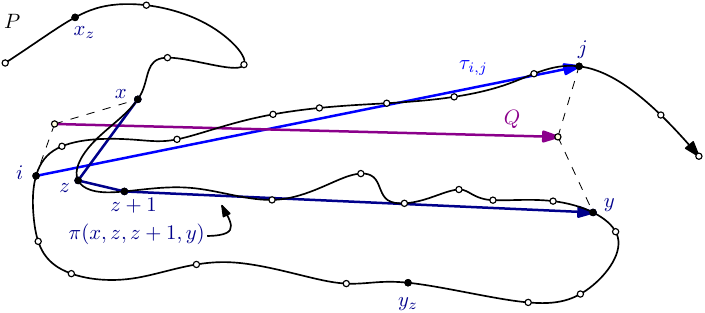}
    \caption{Example of a curve $P$ and index $z$, such that $z\in r_{i,j}$ for some $r_{i,j} \in \ASpace_2$. Also shown is a line segment $Q$, such that $z \in r_Q$ of the initial set system $\RSpace$. After preprocessing, we can test $z \in r_{i,j}$ in constant time. }
    \label{fig:hittingwithlines}
\end{figure}

\subsection{Analysis of the approximation error}
\label{sec:approxerror:line}
In this section we show how we use a set cover of the set system $\ASpace_2$ to construct an approximate solution for our clustering problem and analyse the resulting approximation error. In particular, we prove Lemma~\ref{lem:hittinglines:1} and Lemma~\ref{lem:hittinglines:2}. 

\begin{lemma}\label{lem:hittinglines:1} 
Assume there exists a set cover for $\RSpace$ with parameter $\Delta$.
Let $S$ be a set cover of size $k$ for $\ASpace_2$. We can derive from $S$ a set of $k$ cluster centers $C \subseteq \XX^d_{2}$ and such that $\phi(P,C) \leq 6\Delta$.
\end{lemma}

\begin{proof}
We set $C=\{\tau_{i,j}\;|\;r_{i,j}\in S\}$. Let $r_{i,j}\in S$ and let $z \in r_{i,j}$. By the definition of $r_{i,j}$ there are $x \in [x_z, z]$ and  $y \in [z+1,y_z]$ such that $d_F(\pi({x,z,z+1,y}), \tau_{i,j}) )  \leq 2\Delta$. In the following we show that $d_F(\tau_{i,j}, P[t_{x},t_{y}]) \leq 6 \Delta$. With the triangle inequality we get that $d_F(\tau_{i,j}, P[t_{x},t_{y}])$ is at most the sum of
\[ d_F( \tau_{i,j},\pi(x,z,z+1,y) ) \] and
\[\max( d_F(\tau_{x,z}, P[t_{x},t_z]), d_F(\tau_{z,z+1}, P[t_{z}, t_{z+1}]), d_F(\tau_{z+1,y}, P[t_{z+1}, t_{y}]) ).\] 
 By the choice of $x$ and $y$ we have that
\[ \max( d_F(\tau_{x,z}, P[t_{x},t_z]),  d_F(\tau_{z+1,y}, P[t_{z+1}, t_{y}]) )\leq 4\Delta. \]

It remains to show that $ d_F(\tau_{z,z+1}, P[t_{z}, t_{z+1}])  \leq 4\Delta$. Since there exists a set cover of $\RSpace$ with parameter $\Delta$, there exists a curve $Q \in \XX_2^d$ and $1 \leq i' \leq z \leq z+1 \leq j' \leq m$ such that $d_F(Q,P[t_{i'},t_{j'}]) \leq \Delta$. Therefore there exists $[a,b] \subseteq [0,1]$ such that $d_F(Q[a,b],P[t_{z},t_{z+1}]) \leq \Delta$. Because shortcutting cannot increase the Fréchet distance to a line segment, we also have $d_F(Q[a,b],\tau_{z,z+1}) \leq \Delta$. By triangle inequality it now follows
\[ d_F(\tau_{z,z+1},P[t_{z},t_{z+1}]) \leq d_F(\tau_{z,z+1},Q[a,b]) +  d_F(Q[a,b],P[t_{z},t_{z+1}]) \leq 2\Delta.\]

Since $S$ is a set cover, it holds for the ground set $Z$, that $Z = \bigcup_{(i,j) \in S} r_{i,j}$. Therefore, if we choose $C=\{\tau_{i,j}\;|\;r_{i,j}\in S\}$, then $\phi(P,C) \leq 6\Delta$.
\end{proof}

\begin{lemma}\label{lem:hittinglines:2}
If there exists a set cover $S$ of $\RSpace$, then there exists a set cover of the same size for $\ASpace_2$.
\end{lemma}

To prove this lemma, we first prove the following simple lemma.
\begin{lemma}\label{lem:monotonicity}
Let $1\leq i' \leq i \leq j \leq j'\leq m$ be indices. If $d_F(\tau_{i',j'}, P[t_{i'}, t_{j'}]) \leq \alpha$, then we have
$d_F(\tau_{i,j}, P[t_{i}, t_{j}]) \leq 2\alpha$.
\end{lemma}
\begin{proof}
There exists a line segment $\tau' \subseteq \tau_{i',j'}$, such that $d_F(\tau', P[t_{i},t_{j}])\leq \alpha$. Since shortcutting cannot increase the Fréchet distance to a line segment, we also have $d_F(\tau', \tau_{i,j})\leq \alpha$.
By triangle inequality it now follows that 
\[ d_F(\tau_{i,j}, P[t_{i}, t_{j}]) \leq d_F(\tau_{i,j}, \tau') + d_F(\tau', P[t_{i},t_{j}]) \leq 2\alpha. \]
\end{proof}

\noindent\textit{Proof of Lemma \ref{lem:hittinglines:2}.}
 We claim that for any set $r_{Q} \in \RSpace$  there exists a set $r \in \ASpace_2$,  such that $r_Q \subseteq r$. This claim implies the lemma statement. It remains to prove the claim.

We can rewrite the definition of $r_Q$. Let $Y$ be the set of tuples $(i,j) \in \NN^2$ with $1 \leq i < j \leq m$ and  $d_F(Q,\curveP[t_i,t_j]) \leq \Delta$. We have that
$r_{Q} =  \bigcup_{(i,j) \in Y} [i,j) \cap \NN $.

Let $(i,j) \in Y$. We show that $r_Q \subseteq r_{i,j} \in \ASpace_2$. Let $z \in r_Q$. By the definition of $r_Q$ we have
\[\exists~ x \leq z < y \text{ s.t. }  d_F(Q,\curveP[t_x,t_y]) \leq \Delta. \]

To show that $z \in r_{i,j}$, we prove that the following two conditions hold:
\begin{enumerate}[(i)]
    \item $x \in [x_z,z]$ and $y \in [z+1,y_z]$,
    \item $d_F(\pi({x,z,z+1, y}), \tau_{i,j}) \leq 2\Delta$. 
\end{enumerate}

Since $d_F(Q,P[t_{x},t_{y}]) \leq \Delta$ and shortcutting cannot increase the 
Fréchet-distance to a line segment, we also have \[d_F(Q,\pi(x,z,z+1,y)) \leq \Delta.\]

Similarly, we can conclude $d_F(Q,\tau_{i,j}) \leq \Delta$.
It now follows from the triangle inequality, that
\[ d_F(\pi(x,z,z+1,y), \tau_{i,j}) \leq d_F(\pi(x,z,z+1,y),Q) + d_F(Q, \tau_{i,j} )\leq 2 \Delta. \]
This implies condition (ii). 

The first condition (i) follows in a similar way. Since $r_Q\in S$, there exists a line segment $q_x \subseteq Q$, such that 
$d_F(q_x, P[t_{x},t_{z}])\leq \Delta$. Applying again that shortcutting cannot increase the 
Fréche-distance to a line segment, we also get $d_F(q_x,\tau_{x,z})\leq \Delta$. By the triangle inequality, we have
\[ d_F(\tau_{x,z}, P[t_{x},t_{z}]) \leq d_F(\tau_{x,z},q_x) + d_F(q_x, P[t_{x},t_{z}] )\leq 2 \Delta. \]
Therefore, by Lemma~\ref{lem:monotonicity}, for all $x' \in [x,z]$ 
$d_F(\tau_{x',z}, P[t_{x'}, t_z]) \leq 4\Delta$. As such, $x$ is encountered in the scan and ends up being contained in the interval $[x_z, z]$.

We can make a symmetric argument to show that $d_F(\tau_{z+1,y}, P[t_{z+1},t_{y}])$ and conclude using Lemma~\ref{lem:monotonicity} that $y \in [z+1,y_z]$.
This proves condition (i). 

Together, the above implies that $z \in r_{i,j}$ for $r_{i,j} \in \ASpace_2$. Therefore $r_Q \subseteq r_{i,j}$ for some $r_{i,j} \in \ASpace_2$.
\qed

\subsection{The algorithm}\label{sec:algline}
We intend to use the algorithm of Theorem~\ref{thm:framework} to find a set cover of the set system $\ASpace_2$, since such a set cover gives a $6$-approximation for our clustering problem; see Section~\ref{sec:apendix:framework} for details on the algorithm. The algorithm requires a set system oracle for $\ASpace_2$.
In this section, we describe such a set system oracle. In particular, we show how to build a data structure that answers a query, given indices $i,j$ and $z$, for the predicate $z \in r_{i,j}$ in $O(1)$ time.

\paragraph{The data structure.}
To build the data structure for the oracle, we first compute the indices $x_z$ and $y_z$ for each $1\leq z\leq m-1$, as specified in the definition of the set system in Section~\ref{subsec:setsysline}. 
Next, we construct a data structure that can answer for a pair of breakpoints $i$ and $z$ if there is a breakpoint $x$ with $x_z\leq x\leq z$ such that   $\|P(t_{i})-P(t_{x})\|\leq 2 \Delta$ in $O(1)$ time. For this we build an $m\times m$ matrix $M$ in the following way. For each breakpoint $i$ we go through the sorted list of breakpoints and check if $\|P(t_{i})-P(t_{j})\|\leq 2 \Delta$ for each $1\leq j \leq m$. While doing that, we determine for each $j$  which is the first breakpoint $z_{i,j}\geq j$ with $\|P(t_i)-P(t_{z_{i,j}})\|\leq 2 \Delta$. The entries $z_{i,j}$ are then stored in the matrix $M$ at position $M(i,j)$. Given the Matrix $M$ the oracle can answer if there is a breakpoint $x$ with $x_z\leq x\leq z$ such that   $\|P(t_{i})-P(t_{x})\|\leq 2 \Delta$ by checking if  $M(i,x_z)\leq z$. The data structure can also answer if there is a breakpoint $y$ with $z+1\leq y\leq y_z$ such that   $\|P(t_{j})-P(t_{y})\|\leq 2 \Delta$ by checking if  $M(j,z+1)\leq y_z$. The final data structure stores the matrix $M$ only.

\paragraph{The query.}
We answer queries as follows. Given $z, i$ and $j$, we want to determine if $z \in r_{i,j}$. We return ``yes'', if the following three conditions are satisfied:
\begin{inparaenum}[(i)]
    \item $M(i,x_z)\leq z$
    \item $M(j,z+1)\leq y_z$
    \item $\|s-P(t_z)\| \leq 2\Delta$, where 
     $s$ is the intersection of the bisector between the points $P(t_z)$ and $P(t_{z+1})$ and the line segment $\tau_{i,j}$.  
\end{inparaenum}
Otherwise, the algorithm returns ``no''.

\paragraph{Correctness.}
The above described set system oracle returns the correct answer. Correctness is implied by the following observation, which follows from the analysis of Alt and Godau~\cite{AltG95}. See also Figure~\ref{fig:constant_time}.

\begin{observation}\label{obs:querycases}
$d_F(\pi(x,z,z+1,y), \tau_{i,j}) \leq 2\Delta$ if and only if the following three conditions are satisfied:
\begin{compactenum}[(i)]
    \item $\|P(t_{x})-u\|\leq 2\Delta$
    \item $\|P(t_{y})-v\|\leq 2\Delta$

    \item $ \min_{\genfrac{}{}{0pt}{}{\lambda, \lambda' \in [0,1]}{\lambda \leq \lambda'}} 
    (\| a - (\lambda v + (1-\lambda)u ) \|,
    \| b - (\lambda' v + (1-\lambda') u) \|) \leq 2\Delta $
\end{compactenum}
 where  $a=P(t_z)$, $b=P(t_{z+1})$, $u=P(t_i)$, and $v=P(t_j)$.
\end{observation}  

\begin{figure}
    \centering
    \includegraphics[width=0.39\textwidth]{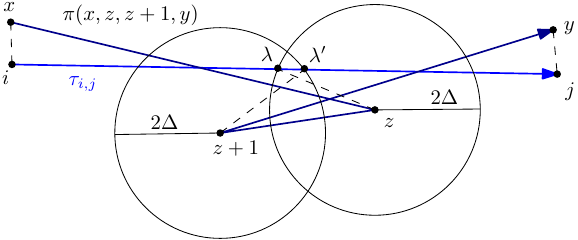}
    \hspace{1cm}
    \includegraphics[width=0.39\textwidth]{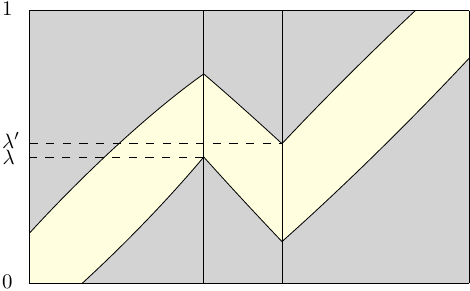}
    \caption{Illustration of Observation~\ref{obs:querycases}. The figure on the right shows the $2\Delta$-free-space diagram of the two curves on the left. A monotone path from the bottom left to the upper right corner of the diagram is feasible iff the three conditions stated in the observation are satisfied. We slightly abuse notation by referring to the vertex $P(t_z)$ with $z$ in all figures, when context is clear.}
    \label{fig:constant_time}
\end{figure}

\paragraph{Running time.} 
Next, we analyse the running time of constructing an oracle for the case $\ell=2$ and query time $O(1)$. In particular we analyse the running time of the scan for the indices $x_z$ (or $y_z$) with $1\leq z<m$ and the running time for building the matrix $M$.

As described above 
the index-scan for $x_z$, given $z$, can be done by checking for breakpoints $x \in \{z-1,\dots,1\}$ in backwards order from $z$ if $d_F(\tau_{x,z},  P[t_{x},t_z]) \leq 4\Delta$.  
Since $\tau_{x,z}$ has complexity $2$ and $P[t_{x},t_z]$ has complexity at most $n$, the check $d_F(\tau_{x,z},  P[t_{x},t_z]) \leq 4\Delta$ can be done in $O(n)$ time and  $O(n)$ space for any $x,z \in \{1,\dots,m\}$ using standard methods  \cite{AltG95}. The scan for $y_z$ is analogous, so we need a total time of $O(mn)$ to scan for all indices.

For building the matrix $M$, the algorithm computes the Euclidean distances of all $\binom{m}{2}$ pairs of breakpoints and while doing that records for each breakpoint $t_j$ the smallest index of a breakpoint after $t_j$ that lies within distance $2\Delta$ to this breakpoint.
In total, this it takes $O(m^2)$ time.
Together with the scan for the indices we get the following runtime for building the oracle.

\begin{restatable}{theorem}{oracled}\label{thm:oracle:d}
One can build a data structure of size $O(m^2)$ in time $O(m(m+n))$ and space $O(n+m^2)$ that answers for an element of the ground set $Z$ and a set of $\ASpace_2$, whether this element is contained in the set in $O(1)$ time.
\end{restatable}

\subsection{The result}
For the set system $(Z,\ASpace_2)$, we have $|Z|=m$ and $|\ASpace_2|=O(m^2)$. Thus, the VC-dimension $\delta$ of the dual set system is trivially bounded by $O(\log m)$.
We combine this with the result for constructing the oracle in Theorem~\ref{thm:oracle:d} and apply Theorem~\ref{thm:framework} to get the following lemma on computing set covers of $\ASpace_2$. Note that we must have $k < m$, since there are only $m-1$ elements in the ground set.

\begin{lemma}
Let $k$ be the minimum size of a set cover for $\ASpace_2$. There exists an algorithm that computes a set cover for $\ASpace_2$ of size $O( k\log^2 (m))$ with an expected running time in 
$ \widetilde{O}\left( k m^2 + mn\right) $ 
and using space in $O(n+m^2)$. 
\end{lemma}

As a direct consequence we get the following result for our clustering problem in the case $\ell=2$ with the help of Lemma~\ref{lem:hittinglines:1}  and Lemma~\ref{lem:hittinglines:2}.

\begin{restatable}{theorem}{mainLines}\label{thm:main:lines}
Let $\curveP: [0,1] \rightarrow \RR^d$ be a polygonal curve of complexity $n$ with breakpoints $0 \leq t_1,\dots,t_m\leq 1$ and let $\Delta > 0$ be a parameter. Assume there exists a set $C^{*} \subset \XX^d_{2}$ of size $k \leq  m$, such that $\phi(P,C^{*})\leq \Delta$. There exists an algorithm that computes a set $C \subset \XX^d_{2}$ of size 
$O( k \log^2 (m))$ such that $\phi(P,C)\leq 6\Delta$. The algorithm has  expected running time in
$ \widetilde{O}\left( k m^2 + mn\right) $
and uses space in $O(n+m^2)$.
\end{restatable}

\section{The main algorithm}\label{sec:main:general}

In this section we extend the scheme described in Section~\ref{sec:alg3} to the case $\ell > 2$. As in the previous section, we only consider the discrete problem, where the input is a polygonal curve with breakpoints. Again, the crucial step is a careful definition of a set system for approximation which allows for an efficient implementation of a set system oracle. The main idea is to replace the edges of the proxy curve $\pi$ from Section~\ref{sec:alg3} by simplifications of the corresponding subcurves. We show that we can do this in a way that ensures that these simplifications are nested in a certain way. This in turn will allow us to build efficient oracle data structures for this set system. We will later show how to use the main elements of this algorithm for the continuous case in Section~\ref{sec:cont}.

\subsection{Simplifications}\label{sec:generatesimplifications}

We begin by introducing the following slightly different notion of simplification.
A curve $\curveQ \in \XX^d_{\ell}$ is an \emph{$(\epsilon,\ell)$-simplification} of a curve $\curveP$ if $Q$ has at most $\ell$ vertices and its Fréchet distance to $\curveP$ is at most $\epsilon$. We call the simplification \emph{vertex-restricted} if $V(\curveQ) \subseteq V(\curveP)$ and the vertices of $Q$ have the same order as in $P$. In this context, we say that a point $p$ of $P$ \emph{corresponds} to an edge $e$ of a vertex-restricted simplification of $P$ if it lies in between the two endpoints of $e$ in $P$.
The main purpose of this section is to  define simplifications $\approxcurvesimpl^+(i,j)$,  $\approxcurvesimpl^-(i,j)$ and $\approxcurvesimpl^\circ(i,i+1)$ for $i,j \in \{1,\dots,m\}$ that we will use in the definition of the set system in the next section. Concretely, the simplifications will be defined as the output of the algorithm by Agarwal et. al.~\cite{Agar05}.
In a nutshell, their algorithm works the following way: Let $P$ be a curve with vertices $p_1,\dots,p_n$. Let $f(\frac{\epsilon}{2})$ denote the minimum number of vertices in a vertex-restricted $(\frac{\epsilon}{2},n)$-simplification of $P$. 
To compute a vertex-restricted $(\epsilon,f(\frac{\epsilon}{2}))$-simplification $P'$ of the curve $P$, the algorithm iteratively adds new vertices to the simplification starting with the first vertex $p_1$ of the curve. In each step it takes the last vertex $p_i$ of the simplification and determines with an exponential search the last integer $j\geq0$ such that $d_F(\overline{p_i p_{i+2^j}},P[p_i,p_{i+2^j}])\leq\epsilon$.  After determining $j$ it finds with a binary search the last integer $r\in [2^{j},2^{j+1}]$ such that $d_F(\overline{p_i p_{i+r}},P[p_i,p_{i+r}])\leq\epsilon$. The algorithm terminates when it reaches $p_n$. 

\begin{figure}
    \centering
    \includegraphics[width=0.7\textwidth]{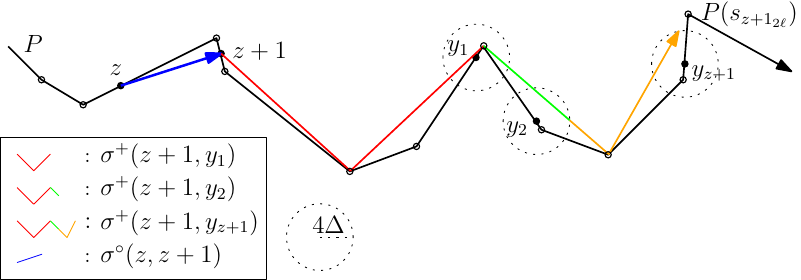}
    \caption{Example of the generated $(4\Delta,2\ell)$-simplifications for a curve $P$ with breakpoints $z, z+1, y_1, y_2$ and $y_{z+1}$ in the case $\ell=2$.}
    \label{fig:generation}
\end{figure}

\paragraph{Generating simplifications.} We now describe how to generate a set of simplifications that will be used in the definition of our set system in Section~\ref{sec:setsystemgeneral}. We apply the above described algorithm on subcurves of $P$ in the following way: For the parameterization $P\colon [0,1] \rightarrow \RR^d$ of $P$ where $P(t_i)$ gives the $i$-th breakpoint of $P$ let $0= s_1\leq \dots \leq s_n= 1$ be the values such that $P(s_j)=p_j$.
For each $z\in\{1,\dots,m\}$ we apply the algorithm with $\epsilon=4\Delta$ on $P[t_z,1]$ to get a simplification $P^+_z$. We stop the algorithm early if the complexity of the simplification reaches $2\ell$. 
If $|P^+_z|=2\ell$ let $P(s_{z_{2\ell}})$ be the $2\ell$-th vertex of $P^+_z$. Otherwise set $P(s_{z_{2\ell}})=p_n$. Let $y_z$ be the last breakpoint of $P$ before $P(s_{z_{2\ell}})$. Let $z\leq y\leq y_z$. Since $P^+_z$ is a  $(4\Delta,2\ell)$-simplification of $P[t_z,1]$, there exists a subcurve $\approxcurvesimpl^+(z,y)$ of $P^+_z$ such that $d_F(\approxcurvesimpl^+(z,y),P[t_z,t_y])\leq4\Delta$. From each possible subcurve with the above property let more specifically $\approxcurvesimpl^+(z,y)$ be the longest subcurve that does not contain any vertex $P(s_i)$ with $s_i\geq t_y$. This subcurve $\approxcurvesimpl^+(z,y)$ is therefore a uniquely defined $(4\Delta,2\ell)$-simplification of $P[t_z,t_y]$ that ends in a point of the edge of $P_z^+$ corresponding to $P(t_y)$. 
Analogously we generate the curve $\approxcurvesimpl^\circ(z,z+1)$ by running the algorithm for the curve $P[t_z,t_{z+1}]$ and the $\approxcurvesimpl^-(x,z)$ by running the algorithm for the direction-inverted curve $P[t_z,0]$. 
We define $P[t_z,0]$ to be the curve $Q:[0,1]\rightarrow \RR^d$ with $Q(t)=P((1-t)t_z)$. Note that it is possible that the algorithm does not find a simplification at all for a specific subcurve. In this case we say the simplification is empty (and we denote this with $\bot$). See also Figure~\ref{fig:generation} for an example of the generated simplifications.

We summarize crucial properties of the generated simplifications in the following two lemmata.  These properties will help to construct an efficient oracle for our set system later. 

\begin{lemma}\label{lem:sigmaplus}
Let $i,j \in \{1,\dots,m\}$ with $i<j$. The curve $\approxcurvesimpl^+(i,j)$ is either a uniquely defined $(4\Delta,2\ell)$-simplification of $P[t_z,t_y]$, or it is $\approxcurvesimpl^+(i,j)=\bot$. In the latter case there exists no $Q\in \XX_\ell^d$ 
such that $d_F(Q,P[t_i,t_j])\leq \Delta$.
Moreover, for any non-empty simplification 
$\approxcurvesimpl^+(i,j)$ and for any $i < j' < j$, the simplification  $\approxcurvesimpl^+(i,j')$ is non-empty and is a subcurve of $\approxcurvesimpl^+(i,j)$.
\end{lemma}

We get symmetric lemmas for the other simplifications. 
We will see in the next section why it is convenient to have these properties in both directions, forwards and backwards along the curve.

\begin{lemma}\label{lem:sigmaminus}
Let $i,j \in \{1,\dots,m\}$ with $i<j$. The curve $\approxcurvesimpl^-(i,j)$ is either a uniquely defined $(4\Delta,2\ell)$-simplification of $P[t_z,t_y]$, or it is  $\approxcurvesimpl^-(i,j)=\bot$. In the latter case there exists no $Q\in \XX_\ell^d$ such that $d_F(Q,P[t_i,t_j])\leq \Delta$. Moreover, for any non-empty simplification 
$\approxcurvesimpl^-(i,j)$ and for any $i < i' < j$ it holds that the simplification $\approxcurvesimpl^-(i',j)$ is non-empty and is a subcurve of $\approxcurvesimpl^-(i,j)$.
\end{lemma}

\begin{lemma}\label{lem:sigma:o}
Let $z \in \{1,\dots,m-1\}$. The curve $\approxcurvesimpl^\circ(z,z+1)$ is either a uniquely defined $(4\Delta,2\ell)$-simplification of $P[t_z,t_{z+1}]$, or it is  $\approxcurvesimpl^\circ(z,z+1)=\bot$. In the latter case there exists no $Q\in \XX_\ell^d$ such that $d_F(Q,P[t_i,t_j])\leq \Delta$. 
\end{lemma}
Lemma~\ref{lem:sigmaplus} follows directly from the following lemma. Lemma~\ref{lem:sigmaminus} and Lemma~\ref{lem:sigma:o} follow by using  symmetric arguments.

\begin{lemma}
Consider the generating process described in Section~\ref{sec:generatesimplifications}.
Let $y$ be a breakpoint of $P$ with $t_y> s_{z_{2\ell}}$. There exists no $Q\in \XX_\ell^d$ such that $d_F(Q,P[t_z,t_y])\leq \Delta$.
\end{lemma}

\begin{proof}
Let $1\leq v\leq n$ such that $s_{v-1}\leq t_y\leq s_v$. So $P(s_v)$ is the first vertex of $P$ after the breakpoint $y$.   Assume there exists a $Q\in \XX_l^d$ such that $d_F(Q,P[t_z,t_y])\leq \Delta$. 

To get a contradiction we will show that, with this assumption, we can construct a vertex-restricted $(2\Delta,2\ell-1)$-simplification of $P[t_z,s_v]$. Let $f(2\Delta)$ denote the minimum number of vertices in a vertex-restricted $(2\Delta,n)$-simplification of $P[t_z,s_v]$. Note that $v>z_{2\ell}$.  So the vertex-restricted $(4\Delta,f(2\Delta))$-simplification $P'$ of the subcurve $P[t_z,s_v]$ computed with the algorithm of Agarwal et. al. has a complexity of at least $2\ell+1$. This follows by the definition of $P(s_{z_{2\ell}})$. Therefore we have $f(2\Delta)\geq 2\ell+1$. But our constructed vertex-restricted $(2\Delta,2\ell-1)$-simplification then would directly contradictict $f(2\Delta)\geq 2\ell+1$.

For the construction of the $(2\Delta,2\ell-1)$-simplification let $\subcurveP= P[t_z,s_{v-1}]$. Since $Q$ is a $(\Delta,\ell)$-simplification of $P[t_z,t_y]$, there exists a subcurve $\tilde{Q}$ of $Q$ with $d_F(\tilde{Q},\subcurveP)\leq \Delta$.
Let $e_1,\dots, e_k$ be the edges of $\tilde{Q}$ and $\tilde{p}_1,\dots,\tilde{p}_j$ be the vertices of $\subcurveP$. It is $k\leq l-1$ and $j\leq n$. Let $\gamma$ be a strictly monotone increasing function such that
\[ d_F(\subcurveP,\tilde{Q}) =  \sup_{t \in [0,1]}
\| \subcurveP(t) - \tilde{Q}(\gamma(t)) \|\leq \Delta. \]
Let further 
\[t_{i_1}= \min\{t\in [0,1]\;|\; \tilde{Q}(\gamma(t))\in e_i, \subcurveP(t)\in \{\tilde{p}_1,\dots,\tilde{p}_j\}\}\]
be the first vertex of $\subcurveP$ that gets mapped to $e_i$ and 
\[t_{i_2}= \max\{t\in [0,1]\;|\; \tilde{Q}(\gamma(t))\in e_i, \subcurveP(t)\in \{\tilde{p}_1,\dots,\tilde{p}_j\}\}\]
be the last vertex of $\subcurveP$ that gets mapped to $e_i$. By construction we have 
\[d_F(\overline{\subcurveP(t_{i_1})\subcurveP(t_{i_2})},\overline{\tilde{Q}(\gamma(t_{i_1}))\tilde{Q}(\gamma(t_{i_2})})\leq \Delta\]
and therefore with the use of triangle inequality
\begin{eqnarray*}
&& d_F(\overline{\subcurveP(t_{i_1})\subcurveP(t_{i_2})},\subcurveP[t_{i_1},t_{i_2}])\\
&\leq& d_F(\overline{\subcurveP(t_{i_1})\subcurveP(t_{i_2})},\overline{\tilde{Q}(\gamma(t_{i_1}))\tilde{Q}(\gamma(t_{i_2})})+d_F(\overline{\tilde{Q}(\gamma(t_{i_1}))\tilde{Q}(\gamma(t_{i_2})},\subcurveP[t_{i_1},t_{i_2}])\\
&\leq& \Delta+\Delta \\
&=& 2\Delta
\end{eqnarray*}
Since $\subcurveP(t_{i_2})$ and $\subcurveP(t_{(i+1)_1})$ are consecutive vertices of $\subcurveP$, we also have
\[d_F(\overline{\subcurveP(t_{i_2})\subcurveP(t_{(i+1)_1})},\subcurveP[t_{i_2},t_{(i+1)_1}])=0.\]
So we can construct a  $(2\Delta,2\ell-1)$-simplification of $P[t_z,s_v]$ by concatenating the vertices 
\[\subcurveP(t_{1_1}),\subcurveP(t_{1_2}),\subcurveP(t_{2_1}),\subcurveP(t_{2_2}),\dots,\subcurveP(t_{k_1}),\subcurveP(t_{k_2}),P(s_{v}).\]
To see that the resulting curve is indeed a vertex-restricted simplification,  we observe that $\subcurveP(t_{1_1})=\subcurveP(0)=P(t_z)$ and that the edge from $\subcurveP(t_{k_2})=P(s_{v-1})$ to $P(s_{v})$ is entirely included in $P$. 
\end{proof}

\subsection{The set system}\label{sec:setsystemgeneral}

We are now ready to define the new set system $\ASpace_3$ with ground set $Z=\{ 1, \dots, m-1 \}$. The set system depends on the simplifications of subcurves of $P$ defined in the previous section. 
Let $(i,j)$ be a tuple with $1\leq i \leq j \leq m$. We say $r_{i,j} = \emptyset$ if there is no  $Q\in \XX_\ell^d$ such that $d_F(Q,P[t_i,t_j])\leq \Delta$.
Otherwise, we define a set $r_{i,j} \in \ASpace_3$ as follows
\[ r_{i,j} = \{ z \in Z \mid  \exists x \in [x_z,z], y \in [z+1,y_{z+1}] \text{ with } d_F(
\kappa_z(x,y), \approxcurvesimpl^+(i,j) )  \leq 10\Delta \},  \]
where
\[\kappa_z(x,y) = \approxcurvesimpl^-(x,z) \oplus \approxcurvesimpl^{\circ}(z,z+1) \oplus \approxcurvesimpl^+(z+1,y)\]
and $x_z \leq z$ is the smallest index such that $\approxcurvesimpl^-(x,z)\neq\bot$ for all $x_z\leq x\leq z$ and $y_{z+1} \geq z+1$ is the highest index such that $\approxcurvesimpl^+(z+1,y)\neq\bot$ for all $z+1\leq y\leq y_{z+1}$. For an example of a curve $P$ with breakpoints $z,i,j$ such that $z\in r_{i,j}$ see  Figure~\ref{fig:newsetsys}.
Note that, by Lemma~\ref{lem:sigma:o} the curve $\approxcurvesimpl^\circ(z,z+1)$ is non-empty for all $z\in\{1,\dots,m-1\}$ if there exists a set of cluster centers $C\subset \XX_\ell^d$ such that $\Phi(P,C)\leq\Delta$. So in this case the set system is well-defined as implied by the Lemmas~\ref{lem:sigmaplus}, \ref{lem:sigmaminus} and \ref{lem:sigma:o}.

\begin{figure}
    \centering
    \includegraphics[width=0.8\textwidth]{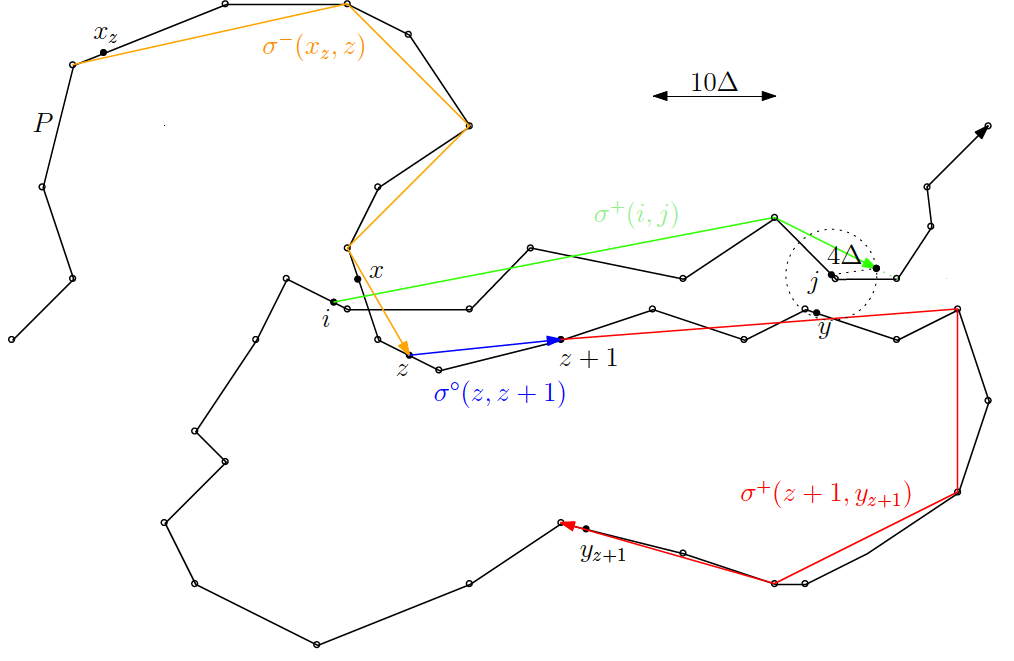}
    \includegraphics[width=0.8\textwidth]{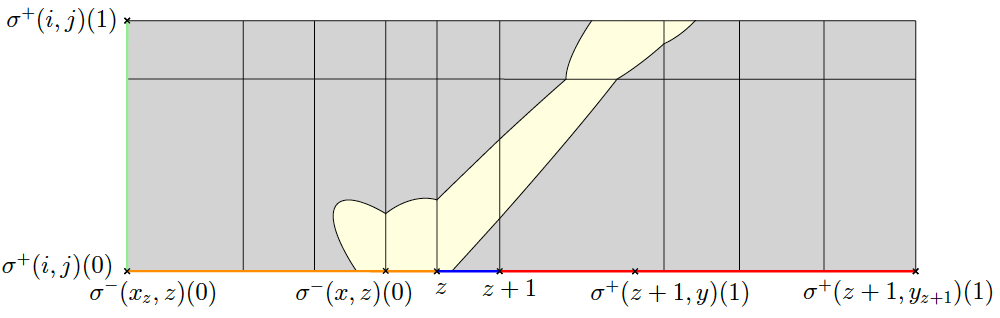}
    \caption{Example of a curve $P$ such that $z\in r_{i,j}$ for some $r_{i,j}\in \ASpace_3$. Also shown is the $10\Delta$-free space diagram of $\kappa_z(x,y)$ and $\approxcurvesimpl^+(i,j)$. Simplification $\approxcurvesimpl^+(i,j)$ demonstrates that the simplifications do not have to be vertex-restricted.}
    \label{fig:newsetsys}
\end{figure}

\subsection{Analysis of the approximation error}

We show correctness in the same schema as in Section \ref{sec:approxerror:line}. In particular, we prove Lemma~\ref{lem:hittingpoly:1} and Lemma~\ref{lem:hittingpoly:2}.

\begin{lemma}\label{lem:hittingpoly:1} 
Let $S$ be a set cover of size $k$ for $\ASpace_3$. We can derive from $S$ a set of $3k$ cluster centers $C \subseteq \XX^d_{\ell}$ and such that $\phi(P,C) \leq 14\Delta$.
\end{lemma}

\begin{proof}
To construct $C$ from $S$ we take  for each tuple $r_{i,j} \in S$ the center curve $\approxcurvesimpl^+(i,j)$. Let $z \in r_{i,j}$. By the definition of $r_{i,j}$ there are $x \in [x_z, z]$ and  $y \in [z+1,y_z]$ such that $d_F(\kappa_z(x,y), \approxcurvesimpl^+(i,j) ) \leq 10\Delta$. In the following we show that $d_F(\approxcurvesimpl^+(i,j) , P[t_{x},t_{y}]) \leq 14 \Delta$. With the triangle inequality we get
\begin{align*}
   d_F(\approxcurvesimpl^+(i,j), P[t_{x},t_{y}]) &\leq d_F(\approxcurvesimpl^+(i,j),\kappa_z(x,y))+ d_F(\kappa_z(x,y),P[t_{x},t_{y}])\\&\leq 10 \Delta +d_F(\kappa_z(x,y),P[t_{x},t_{y}])
\end{align*}
It remains to show that \[d_F(\kappa_z(x,y),P[t_{x},t_{y}])\leq 4 \Delta.\] This follows directly because the distance $d_F(\kappa_z(x,y),P[t_{x},t_{y}])$ is at most the maximum of the distances
$d_F(\approxcurvesimpl^-(x,z),P[t_{x},t_{z}])$, $d_F(\approxcurvesimpl^\circ(z,z+1),P[t_{z},t_{z+1}])$ and $d_F(\approxcurvesimpl^+(z+1,y),P[t_{z+1},t_{y}]))$. We use here that $\approxcurvesimpl^-(x,z)$, $\approxcurvesimpl^\circ(z,z+1)$ and $\approxcurvesimpl^+(z+1,y)$ are $(4\Delta,2\ell)$-simplifications of the corresponding subcurves. Since $S$ is a set cover, it holds for the ground set $Z=\{1,\dots,m-1\}$, that $Z = \bigcup_{(i,j) \in S} r_{i,j}$. Therefore, if we choose $C'=\{\approxcurvesimpl^+(i,j)\;|\;r_{i,j}\in S\}$, we get $\phi(P,C') \leq 14\Delta$. Note that $C'\subseteq \XX^d_{2\ell}$. Let $c\in C'$ with vertices $c_1,\dots, c_N$ where $N\leq 2\ell$. We can split $c$ into 3 curves $c^{(1)},c^{(2)},c^{(3)}$ of complexity at most $\ell$, where $c^{(1)}$ is defined by the vertices $c_1,\dots,c_{\min(N,\ell)}$, the curve $c^{(2)}$ is defined by the vertices $c_{\min(N,\ell)},\dots,c_{\min(N,2\ell-1)}$ and  the curve $c^{(3)}$ is defined by the vertices $c_{\min(N,2\ell-1)},\dots,c_{\min(N,2\ell)}$.   If we split each curve $c\in C'$ as described above, we obtain a set $C\subseteq \XX^d_{\ell}$ with $|C|=3|C'|$ and $\phi(P,C) \leq 14\Delta$.  
\end{proof}

\begin{lemma}\label{lem:hittingpoly:2}
If there exists a set cover $S$ of $\RSpace$, then there exists a set cover of the same size for $\ASpace_3$.
\end{lemma}

\begin{proof}
 We claim that for any set $r_{Q} \in \RSpace$  there exists a set $r \in \ASpace_3$,  such that $r_Q \subseteq r$. This claim implies the lemma statement. It remains to prove the claim.

We can rewrite the definition of $r_Q$. Let $Y$ be the set of tuples $(i,j) \in \NN^2$ with $1 \leq i < j \leq m$ and  $d_F(Q,\curveP[t_i,t_j]) \leq \Delta$. We have that
$r_{Q} =  \bigcup_{(i,j) \in Y} [i,j) \cap \NN $.

Let $(i,j) \in Y$. We show that $r_Q \subseteq r_{i,j} \in \ASpace_3$. Let $z \in r_Q$. By the definition of $r_Q$ we have
\[\exists~ x \leq z < y \text{ s.t. }  d_F(Q,\curveP[t_x,t_y]) \leq \Delta \]

To show that $z \in r_{i,j}$, we prove that the following two conditions hold:
\begin{enumerate}[(i)]
    \item $x \in [x_z,z]$ and $y \in [z+1,y_z]$,
    \item $d_F(\kappa_z(x,y), \approxcurvesimpl^+(i,j) ) \leq 10\Delta$. 
\end{enumerate}
As stated above, we have $d_F(Q,P[t_{x},t_{y}]) \leq \Delta$. Therefore we can subdivide $Q$ into 3 subcurves $Q_x,Q_z,Q_y$ such that \[\max(d_F(Q_x,P[t_{x},t_{z}]),d_F(Q_z,P[t_{z},t_{z+1}]),d_F(Q_y,P[t_{z+1},t_{y}])) \leq \Delta\] 
Each of the subcurves has complexity at most $\ell$ since $Q$ has complexity at most $\ell$. By the Lemmas~\ref{lem:sigmaminus} and \ref{lem:sigmaplus},  we have  $\approxcurvesimpl^-(x',z)\neq\bot$ for all $x\leq x'\leq z$ and  $\approxcurvesimpl^+(z+1,y')\neq\bot$ for all $z+1\leq y'\leq y_{z+1}$.
We can conclude that $x \in [x_z,z]$ and $y \in [z+1,y_z]$ and therefore condition (i) is fulfilled.

To prove condition (ii) we can use the triangle inequality to get
\[   d_F(\kappa_z(x,y), \approxcurvesimpl^+(i,j) ) \leq d_F(\kappa_z(x,y), Q)+ d_F(Q, \approxcurvesimpl^+(i,j))  \]
Since we have
\begin{align*}
    d_F(\kappa_z(x,y), Q) &\leq d_F(\kappa_z(x,y), P[t_{x},t_{y}])+ d_F(P[t_{x},t_{y}], Q)\\
    &\leq 4\Delta +\Delta\\
    &= 5\Delta
\end{align*}
and
\begin{align*}
    d_F(Q,\approxcurvesimpl^+(i,j)) &\leq d_F(Q, P[t_{i},t_{j}])+ d_F(P[t_{i},t_{j}], \approxcurvesimpl^+(i,j))\\
    &\leq \Delta +4\Delta\\
    &= 5\Delta
\end{align*}
we get in total
\[  d_F(\kappa_z(x,y), \approxcurvesimpl^+(i,j) )\leq 10\Delta\]
Together, the above implies that $z \in r_{i,j}$ and therefore $r_Q \subseteq r_{i,j}$.
\end{proof}

\subsection{The approximation oracle}\label{sec:approxorac}
To find a set cover of the set system $\ASpace_3$ we want to use the framework described in Section~\ref{sec:genfram}. But to apply Theorem~\ref{thm:framework} directly we would need to implement an oracle that answers for an element of the ground set $Z = \{1,\dots,m-1\}$ and a set of $\ASpace_3$, whether this element is contained in the set. In this section we describe how to answer such queries approximately. In the next section (Section~\ref{sec:application}) we then show how to apply Theorem~\ref{thm:framework}.

The approximation oracle will have the following properties. Given a set $r_{i,j} \in \ASpace_3$
and an element $z \in Z$ this approximation oracle returns either one of the following answers:
\begin{compactenum}[(i)]
\item "Yes", in this case there exists $x \in [x_z,z]$ and $ y \in [z+1,y_{z+1}]$ with $d_F(
\kappa_z(x,y), \approxcurvesimpl^+(i,j) )  \leq 46\Delta$
\item  "No", in this case $(i,j)\notin r_z$.
\end{compactenum}
In both cases the answer is correct.

To construct the approximation oracle we build a data structure that answers a query, given indices $i$,$j$ and $z$, for the predicate $z\in r_{i,j}$ in $O(\ell^2)$ time. In particular we need a data structure that can build a free space diagram of the curves $\kappa_z(x_z,y_{z+1})$ and $\approxcurvesimpl^+(i,j)$ to bound the distance $d_F(\kappa_z(x,y), \approxcurvesimpl^+(i,j) )$ for every $x\in[x_z,z]$ and $y\in[z+1,y_{z+1}]$. In this context we define active edges of the simplifications $\approxcurvesimpl^-(x_z,z)$ and $\approxcurvesimpl^+(z+1,y_{z+1})$ with respect to $r_{i,j}$ since the data structure needs to be able to find these efficiently to answer the query. Recall that a point of $P$ is said to correspond to an edge $e$ of a vertex-restricted simplification of $P$ if it lies in between the two endpoints of $e$ in $P$.
\begin{definition}Let $z,i,j$ be breakpoints of $P$.
An edge $e$ of the simplification $\approxcurvesimpl^-(x_z,z)$ is \emph{active} with respect to $r_{i,j}$ if
there is a breakpoint $x\in[x_z,z]$ corresponding to $e$ with $d(P(t_x),P(t_i))\leq 18 \Delta$. 
An edge $e$ of the simplification $\approxcurvesimpl^+(z+1,y_{z+1})$ is \emph{active} with respect to $r_{i,j}$ if there is a breakpoint $y\in[z+1,y_{z+1}]$  corresponding to $e$ with $d(P(t_y),P(t_j))\leq 18 \Delta$. 
\end{definition}

\begin{figure}
    \centering
    \includegraphics[width=0.8\textwidth]{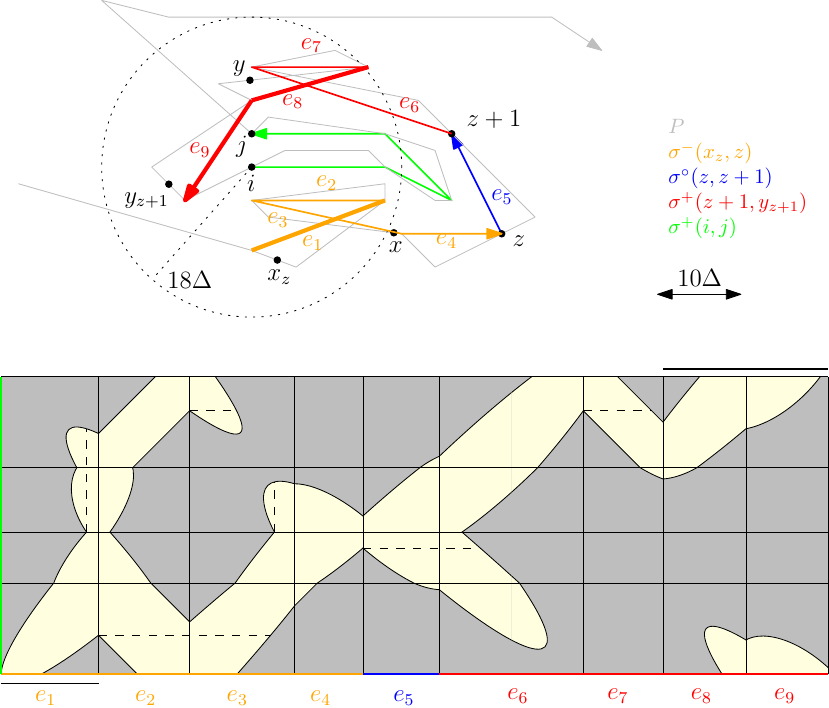}
    \caption{Example of a curve $P$ and breakpoints $x_z$, $x$, $z$, $z+1$, $y$, $y_{z+1}$, $i$ and $j$. The active edges are $e_1$, $e_8$ and $e_9$ since there are breakpoints corresponding to these edges within distance $18\Delta$ to $i$ or $j$ respectively. There is, however, no strictly monotone path from $e_1$ on the bottom to $e_8$ or $e_9$ on the top in the $10\Delta$-free space of $\approxcurvesimpl^+(i,j)$ and $\kappa_z(x_z,y_{z+1})$. So we have $z\notin r_{i,j}$.} 
    \label{fig:activeedge}
\end{figure}

So an active edge is an edge of the simplification that contains the image of a breakpoint that is close to $i$ or $j$ respectively. The active edges will become relevant for answering a query since in the case that $z\in r_{i,j}$ there exist breakpoints $x$ and $y$ on active edges such that $d_F(\kappa_z(x,y), \approxcurvesimpl^+(i,j) )\leq 10\Delta$. For an approximate solution it will suffice to check the existence of a strictly monotone path in the free space diagram that start on an active edge of $\approxcurvesimpl^-(x_z,z)$ and end in an active edge of $\approxcurvesimpl^+(z+1,y_{z+1})$. The advantage is that this can be done faster than checking if $d_F(\kappa_z(x,y), \approxcurvesimpl^+(i,j) )\leq 10\Delta$ for each $x\in[x_z,z]$ and $y\in[z+1,y_{z+1}]$. See Figure~\ref{fig:activeedge} for an example.

\paragraph{The query.}
Given $z,i,j\in\{1,\dots,m\}$ the oracle is therefore checking if $z\in r_{i,j}$ the following way:

First it builds a free space diagram of $\approxcurvesimpl^+(i,j)$ and $\kappa_z(x_z,y_{z+1})$ for the distance $10 \Delta$. Then it checks for each edge on  $\approxcurvesimpl^-(x_z,z)$ and on $\approxcurvesimpl^+(z+1,y_{z+1})$ if it is active. In the end, the oracle checks if there is a monotone increasing path in the $10\Delta$-free space that starts on an active edge of $\approxcurvesimpl^-(x_z,z)$ in one coordinate and  $\approxcurvesimpl^+(i,j)(0)$ in the other coordinate and ends on  an active edge of $\approxcurvesimpl^+(z+1,y_{z+1})$ in one coordinate and $\approxcurvesimpl^+(i,j)(1)$ in the other coordinate. The oracle returns "Yes" if such a path exists. See Figure \ref{fig:pathfrs} for an example of a "Yes" answer.

To do the above steps efficiently an underlying data structure for the oracle has to be built in the preprocessing. We  will first show how the data structure is built and then prove the correctness of the oracle and analyse its running time.

\paragraph{The data structure.}
The data structure is built in two steps. The first step is to compute the simplifications. The second step consists of constructing a data structure for the breakpoints that can be used to determine active edges.

We compute the simplifications $\approxcurvesimpl^-(x_z,z)$, $\approxcurvesimpl^\circ(z,z+1)$ and $\approxcurvesimpl^+(z,y_z)$ for every breakpoint $z \in \{1,\dots,m\}$ by running the algorithm of Agarwal et. al. \cite{Agar05} up to complexity $2\ell$. For each edge $e$ of  $\approxcurvesimpl^-(x_z,z)$ and $\approxcurvesimpl^+(z,y_z)$, we save the first breakpoint $x_e$ and the last breakpoint $y_e$ that corresponds to $e$.
    
In addition to these simplifications, the oracle also needs the simplification $\approxcurvesimpl^+(i,j)$ to build the free space diagram. Note that $\approxcurvesimpl^+(i,j)$ does not need to be stored in the data structure since for all $i,j\in \{1,\dots,m\}$, the  simplification $\approxcurvesimpl^+(i,j)$ can be constructed using $\approxcurvesimpl^+(i,j_i)$. To do so, the oracle does binary search to find the edge $e$ of $\approxcurvesimpl^+(i,j_i)$ such that $j$ corresponds to $e$. Then, the oracle computes the last point of $e$ that intersects the ball $B(t_j,4\Delta)$. The subcurve of $\approxcurvesimpl^+(i,j_i)$ up to this point is $\approxcurvesimpl^+(i,j)$.

The oracle needs to determine which edges are active. For this we construct a data structure in the same way as described for the case $\ell=2$ in Section~\ref{sec:algline}. We build an $m\times m$ matrix $M$ which  stores the following information. For each breakpoint $i$ we go through the sorted list of breakpoints and check if $d(P(t_i),P(t_j))\leq 18 \Delta$ for each $1\leq j \leq m$. While doing that, we determine for each $j$  which is the first breakpoint $z_{i,j}\geq j$ with $d(P(t_i),P(t_j))\leq 18 \Delta$. The entries $z_{i,j}$ are then stored in the matrix $M$.
    
Let $x_e (y_e)$ be the first (last) breakpoint corresponding to the edge $e$. To check if there is one breakpoint $z$ on an edge $e$ of a simplification such that   $d(P(t_i),P(t_z))\leq 18 \Delta$ for some other breakpoint $i$, we only have to check if $z_{i,x_e}\geq y_e$. This is exactly what we need to check to decide if an edge is active and can be done in constant time given the matrix $M$.
    
Overall, the data structure therefore consists of $O(m)$ simplifications with pointers to the first (last) element of each edge and the matrix $M$ of size $O(m^2)$ containing the $z_{i,j}$-entries. This data structure is then used for each query to build a free space diagram and to find the active edges. The existence of a monotone increasing path is then tested by computing the reachability of active edges from active edges in the free space diagram. This can be done using the standard methods described by Alt and Godau~\cite{AltG95} in the following way.

The free space diagram of the $10\Delta$-free space $F$ can be divided into cells that each correspond to a pair of edges, one from each curve  $\kappa_z(x_z,y_{z+1})$ and $\approxcurvesimpl^+(i,j)$. Let us denote with $C_{s,t}$ the cell of the free space diagram corresponding to the $s$-th edge of $\approxcurvesimpl^+(i,j)$ and the $t$-th edge $e_t$ of $\kappa_z(x_z,y_{z+1})$. We further denote with $L_{s,t}$ and $B_{s,t}$ the left and bottom line segment bounding the cell $C_{s,t}$. We also define $L_{s,t}^F=L_{s,t}\cap F$ and $B_{s,t}^F=B_{s,t}\cap F$. 

We need to calculate the reachable space $R\subseteq F$ where a point $p\in F$ is in $R$ if and only if there exists an active edge $e_t$ of $\approxcurvesimpl^-(x_z,z)$ such that there exists a monotone increasing path within $F$ from $B_{1,t}^F$ to $p$. We further define $L_{s,t}^R=L_{s,t}\cap R$ and $B_{s,t}^R=B_{s,t}\cap R$.

Note that given $L_{s,t}^R$, $B_{s,t}^R$,$L_{s+1,t}^F$ and $B_{s,t+1}^F$, we can construct $L_{s+1,t}^R$ and $L_{s,t+1}^R$ in constant time. So, given that we know for each edge $e_t$ of $\approxcurvesimpl^-(x_z,z)$, whether it is active or not, we can compute $L_{1,t}^R$ and $B_{1,t}^R$ for all edges $e_t$. With these we can iteratively construct all $L_{s,t}^R$ and $B_{s,t}^R$, proceeding row by row in the free space diagram. 

Let $s^*\leq 2\ell$ be the number of edges of $\approxcurvesimpl^+(i,j)$. We get the following directly from the definition of $R$. There exists an active edge $e_t$ of $\approxcurvesimpl^+(z+1,y_{z+1})$ such that $B_{s^*+1,t}^R\neq \emptyset$ if and only if there is a monotone increasing path starting and ending in an active edge. So we only have to check for all active edges $e_t$ of $\approxcurvesimpl^+(z+1,y_{z+1})$ if $B_{s^*+1,t}^R\neq \emptyset$.

\paragraph{Correctness.}To show the correctness of the oracle we show the following lemma.
\begin{lemma}
Let $z,i,j\in\{1,\dots,m\}$. Consider the query $z\in r_{i,j}$. If the approximation oracle returns the answer 
\begin{compactenum}[(i)]
\item "Yes", then there exists $x \in [x_z,z]$ and $ y \in [z+1,y_{z+1}]$ with $d_F(
\kappa_z(x,y), \approxcurvesimpl^+(i,j) )  \leq 46\Delta$
\item  "No", then we have  $z \notin r_{i,j}$.
\end{compactenum}
\label{lem:correctapror}
\end{lemma}

 \begin{figure}
    \centering
    \includegraphics[width=0.8\textwidth]{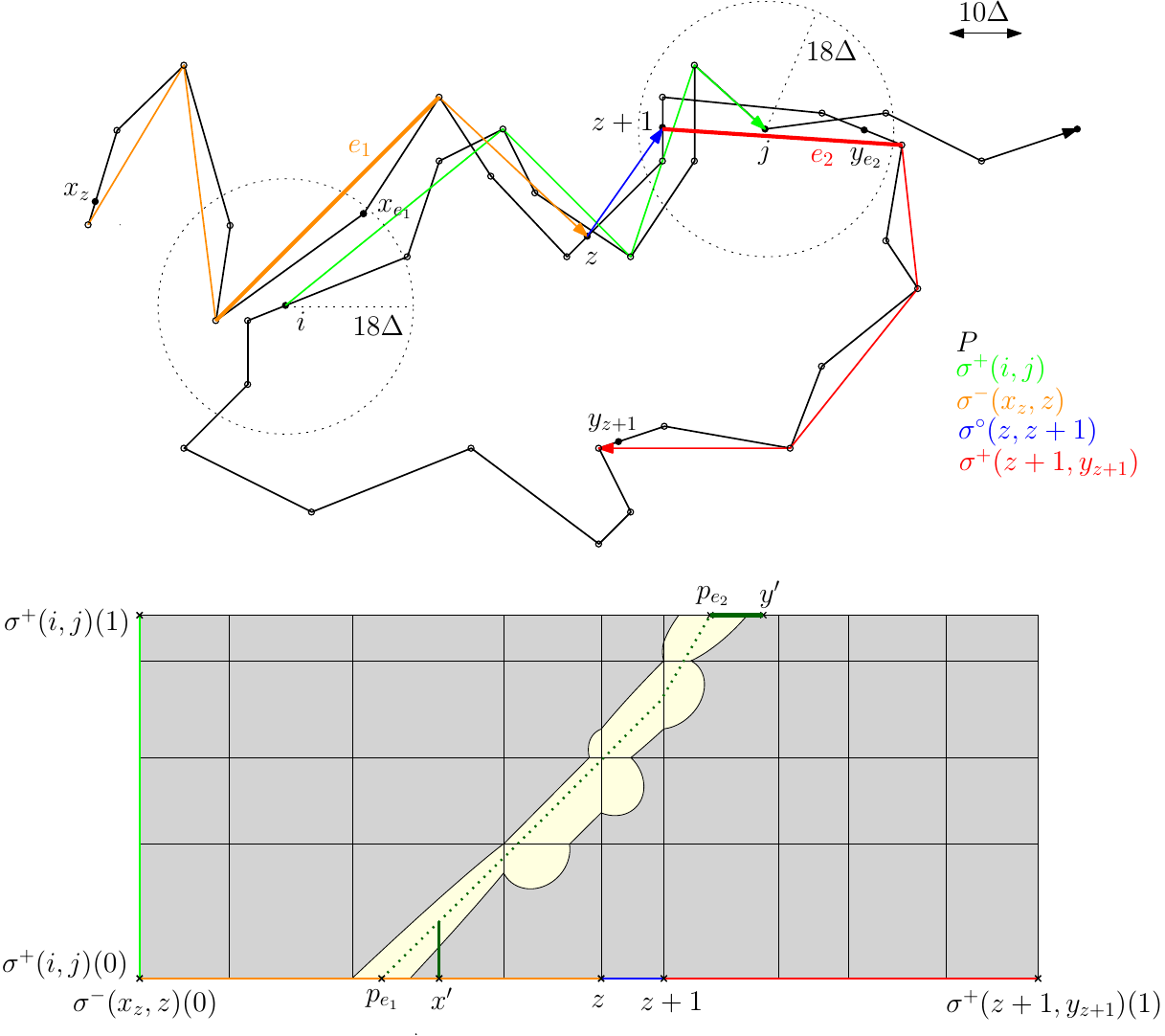}
    \caption{ Example for a curve $P$ and breakpoints $z,i,j$ such that the approximation oracle returns "Yes" for the query $z\in r_{i,j}$. The path from $p_{e_1}$ to $p_{e_2}$ in the $10\Delta$-free space diagram is a monotone increasing path from the active edge $e_1$ to the active edge $e_2$. The edges are active since $d(P(t_i),P(t_{x_{e_1}}))\leq 18\Delta$ and $d(P(t_j),P(t_{y_{e_2}}))\leq 18\Delta$. The path from $x'=\approxcurvesimpl^-(x_{e_1},z)(0)$ to $y'=\approxcurvesimpl^+(z+1,y_{e_2})(1)$ in the free space diagram  gives a parametrization of $\kappa_z(x_{e_1},y_{e_2})$ and $\approxcurvesimpl^+(i,j) $ yielding $d_F(\kappa_z(x_{e_1},y_{e_2}), \approxcurvesimpl^+(i,j)) \leq 46\Delta$ as proven in Lemma \ref{lem:correctapror}.}
    \label{fig:pathfrs}
\end{figure}

\begin{proof}
i) Consider the $10\Delta$-free space diagram of $\kappa_z(x,y)$ and $\approxcurvesimpl^+(i,j) $. If the oracle returns the answer "Yes" then there is a monotone increasing path in the $10\Delta$-free space that starts on an active edge $e$ and ends on an active edge $e'$. 

We show that this path implicitly gives two breakpoints $x_e \in [x_z,z]$ and $y_{e'} \in [z+1,y_{z+1}]$ as well as a monotone increasing path from $\approxcurvesimpl^-(x_e,z)(0)$ to $\approxcurvesimpl^+(z+1,y_{e'})(1)$ in the $46\Delta$-free space of $\kappa_z(x,y)$ and $\approxcurvesimpl^+(i,j) $.

 Let $x_e$ be the first breakpoint corresponding to $e$ such that $d(P(t_{x_e}),P(t_i))\leq 18 \Delta$. Since $e$ is active, $x_e$ has to exist. We distinguish between the cases that the path starts in a point $p_e$ before or after $\approxcurvesimpl^-(x_e,z)(0)$ on $e$:
\begin{itemize}[$\bullet$]
    \item[(I)] The path starts in a point $p_e$ after $\approxcurvesimpl^-(x_e,z)(0)$ on $e$: 

We have
\begin{eqnarray*}
&& d(\approxcurvesimpl^+(i,j)(0),\approxcurvesimpl^-(x_e,z)(0))\\
&\leq& d(\approxcurvesimpl^+(i,j)(0),P(t_i))+d(P(t_i),P(t_{x_e}))+d(P(t_{x_e}),\approxcurvesimpl^-(x_e,z)(0))\\
&\leq& 4\Delta+18\Delta+4\Delta\\
&\leq& 26\Delta
\end{eqnarray*}
The second inequality above follows by the choice of $x_e$ and the fact that $\approxcurvesimpl^+(i,j)$ and $\approxcurvesimpl^-(x_e,z)$ are $(4\Delta,2\ell)$-simplifications of $P[t_i,t_j]$ and $P[t_{x_e},t_z]$.
Since the path starts in a reachable area of the free space diagram we have \[d(\approxcurvesimpl^+(i,j)(0),p_e)\leq 10\Delta\]
Since $p_e$ and $\approxcurvesimpl^-(x_e,z)(0)$ lie on the same edge of $\approxcurvesimpl^-(x_z,z)$ the segment $\overline{p_e,\approxcurvesimpl^-(x_e,z)(0)}$ is a subcurve of $\approxcurvesimpl^-(x_z,z)$. The Fr\'echet distance 
\[d_F(\overline{p_e,\approxcurvesimpl^-(x_e,z)(0)},\approxcurvesimpl^+(i,j)(0))\]
is at most
\[
\max(d(\approxcurvesimpl^+(i,j)(0),\approxcurvesimpl^-(x_e,z)(0)), d(\approxcurvesimpl^+(i,j)(0),p_e))\leq 26\Delta\]
 since the Fr\'echet distance of a line segment and a point is attained at the start or end point of the line segment. The horizontal line segment from the point $(p_e,\approxcurvesimpl^+(i,j)(0))$ to the point $(\approxcurvesimpl^-(x_e,z)(0),\approxcurvesimpl^+(i,j)(0))$ is therefore contained in the  $46\Delta$-free space of $\kappa_z(x,y)$ and $\approxcurvesimpl^+(i,j) $.

\item[(II)] The path starts in a point $p_e$ before $\approxcurvesimpl^-(x_e,z)(0)$ on $e$:

We again have \[d(\approxcurvesimpl^+(i,j)(0),\approxcurvesimpl^-(x_e,z)(0))\leq 26\Delta\]
and  \[d(\approxcurvesimpl^+(i,j)(0),p_e)\leq 10\Delta\]
Therefore we have
\begin{align*}
d(p_e,\approxcurvesimpl^-(x_e,z)(0))&\leq d(p_e,\approxcurvesimpl^+(i,j)(0))+d(\approxcurvesimpl^+(i,j)(0),\approxcurvesimpl^-(x_e,z)(0))\\
&\leq 10 \Delta + 26\Delta\\
&\leq 36\Delta
\end{align*}
The path has to  pass the vertical line  in the free space diagram through $\approxcurvesimpl^-(x_e,z)(0)$ at some height $h$. Note that the path is totally included in the $10\Delta$-free space. So  for each point $p$ on $\approxcurvesimpl^+(i,j)[0,h]$ there is a point $q$ on $\approxcurvesimpl^-(x_z,z)$ between $p_e$ and $\approxcurvesimpl^-(x_e,z)(0)$ such that \[d(p,q)\leq 10 \Delta.\] Because $q$ lies on the same edge of $\approxcurvesimpl^-(x_z,z)$ as $\approxcurvesimpl^-(x_e,z)(0)$ and $p_e$ we have
\[d(\approxcurvesimpl^-(x_e,z)(0),q)\leq d(\approxcurvesimpl^-(x_e,z)(0),p_e)\leq 36\Delta\]
and therefore
\begin{align*}
d(\approxcurvesimpl^-(x_e,z)(0),p)&\leq d(\approxcurvesimpl^-(x_e,z)(0),q)+d(q,p)\\
&\leq 36\Delta+10\Delta\\
&\leq 46\Delta
\end{align*}
So we can replace the path in the $10\Delta$-free space starting at $p_e$ up to height $h$ with a vertical line segment from $(\approxcurvesimpl^-(x_e,z)(0),\approxcurvesimpl^+(i,j)(0))$ up to height $h$. This line segment is then fully contained in the $46\Delta$-free space.
\end{itemize}
By symmetry, we can apply the same arguments for changing the path in the free space diagram, so that the path ends in $\approxcurvesimpl^-(z+1,y)(1)$ for some breakpoint $y$. Therefore we can always find a monotone increasing path from $\approxcurvesimpl^-(x_e,z)(0)$ to $\approxcurvesimpl^+(z+1,y_{e'})(1)$ in the $46\Delta$-free space of $\kappa_z(x,y)$ and $\approxcurvesimpl^+(i,j) $. For an example of such a path see Figure \ref{fig:pathfrs}. The vertical path starting in $x'$ is an example for Case II and the horizontal path from $p_{e_2}$ to $y'$ is an example for Case I (by symmetry for the end of the path).

ii) We prove that the oracle returns the answer "Yes" if $z \in r_{i,j}$:

So let $z\in r_{i,j}$ Then we have $d_F(
\kappa_z(x,y), \approxcurvesimpl^+(i,j) )  \leq 10\Delta$ for some $x_z\leq x\leq z$ and $z+1\leq y\leq y_{z+1}$. Therefore there is a path in the free space diagram from $(\approxcurvesimpl^+(i,j)(0),\approxcurvesimpl^-(x,z)(0)))$ to $((\approxcurvesimpl^+(i,j)(1),\approxcurvesimpl^+(z+1,y)(1)))$. It remains to show that the edges corresponding to $x$ and $y$ are active. This follows by triangle inequality. In particular we have that $d(P(t_i),P(t_{x}))$ is at most
\begin{align*}
d(P(t_i),\approxcurvesimpl^+(i,j)(0))+d(\approxcurvesimpl^+(i,j)(0),\approxcurvesimpl^-(x,z)(0))+d(\approxcurvesimpl^-(x,z)(0),P(t_{x})) 
\end{align*}
and by the above this is at most $18\Delta$, and analogously $d(P(t_j),P(t_{y}))\leq 18\Delta$.
\end{proof}

\paragraph{Running time.} First we analyse the preprocessing time needed to build the data structure for the oracle then we analyse the query time of the oracle.

Since one application of the algorithm of Agarwal et. al. \cite{Agar05} needs $O(n\log(n))$ time and $O(n)$ space, we need $O(mn\log(n))$ time and $O(n+m\ell)$ space to construct the simplifications $\approxcurvesimpl^-(x_z,z)$, $\approxcurvesimpl^\circ(z,z+1)$ and $\approxcurvesimpl^+(z,y_z)$ for every $z \in \{1,\dots,m\}$. To construct the pointers from each edge to the first and last breakpoint on the edge we need additional $O(m+\ell)$ time for each simplification. In total this needs at most $O(mn\log(n)+m^2)$ time and $O(n+m\ell)$ space.

To construct the matrix $M$ with the $O(m^2)$ entries of $z_{i,j}$ we need for each breakpoint $i$ a time of $O(m)$ and a space of $O(m)$ to go through the list of all $m$ breakpoints and save the entries of $z_{i,j}$. So in total we need $O(m^2)$ time and $O(m^2)$ space for all entries. Combined with the time and space requirement for the simplifications  we need $O(m(n\log(n)+m+\ell))$ time and  $O(m\ell+m^2)$ space for the whole preprocessing.

To answer a query the oracle builds a free space diagram of $\approxcurvesimpl^+(i,j)$ and $\kappa_z(x_z,y_{z+1})$. To do that, it needs the simplifications $\approxcurvesimpl^+(i,j)$, $\approxcurvesimpl^-(x_z,z)$, $\approxcurvesimpl^\circ(z,z+1)$ and $\approxcurvesimpl^+(z+1,y_z)$. The simplifications $\approxcurvesimpl^-(x_z,z)$, $\approxcurvesimpl^\circ(z,z+1)$ and $\approxcurvesimpl^+(z+1,y_z)$ were already computed during preprocessing. The simplification $\approxcurvesimpl^+(i,j)$ can be computed in $O(\log(l))$ time with binary searches on $\approxcurvesimpl^+(i,y_i)$ and $\approxcurvesimpl^+(z,y_z)$. With the matrix $M$, it can be checked if an edge of $\approxcurvesimpl^-(x_z,z)$ or $\approxcurvesimpl^+(z+1,y_z)$ is active in $O(1)$ time. Therefore all active edges can be found in $O(\ell)$ time. The construction of the free space diagram of two curves with complexity $O(\ell)$ can then be done with standard methods as described earlier in $O(\ell^2)$ time. Testing the existence of a monotone increasing path from any of the active edges is then done as described above in the paragraph about the data structure.  Note that given $L_{s,t}^R$, $B_{s,t}^R$,$L_{s+1,t}^F$ and $B_{s,t+1}^F$, we can  construct $L_{s+1,t}^R$ and $L_{s,t+1}^R$ in $O(1)$ time. Therefore, given that we know for each edge $e_t$ of $\approxcurvesimpl^-(x_z,z)$ if it is active, we can compute $L_{1,t}^R$ and $B_{1,t}^R$ for all edges $e_t$ in $O(\ell)$ time. So we can compute all $L_{s,t}^R$ and $B_{s,t}^R$ in $O(\ell^2)$ time. Since $\approxcurvesimpl^+(z+1,y_{z+1})$ has at most $2\ell$ edges, the check for each of the active edges $e_t$ of $\approxcurvesimpl^+(z+1,y_{z+1})$ if $B_{s^*+1,t}^R\neq \emptyset$ can then be done in $O(\ell)$ time. This implies that testing if there exists a monotone increasing path with the described properties can be done in $O(\ell^2)$ time. Therefore the total query time is $O(\ell^2)$, as well. These results for the running time imply the following theorem.

\begin{theorem}\label{thm:oracle:l1}
One can build a data structure for the approximation oracle of size $O(m\ell+ m^2)$ in time $O\left(m^2+m n\log(n)\right)$ and  space $O(n+m\ell+m^2)$ that has a query time of $O(\ell^2)$. 
\end{theorem}

\subsection{Applying the framework for computing a set cover}\label{sec:application}
In order to apply Theorem~\ref{thm:framework} directly, we technically need to define a set system based on our data structure. Concretely, we define a new set system that is implicitly given by the approximation oracle.
Let $I(z,(i,j))$ be the output of the approximation oracle for $z \in Z $ and $(i,j)\in T$ with
\begin{align*}
    I(z,(i,j))&=1 \;\;\;\;\text{ if the oracle answers "Yes"}\\
    I(z,(i,j))&=0 \;\;\;\;\text{ if the oracle answers "No" } 
\end{align*}
Let ${\ASpace}_{4}$ be the set system consisting of sets of the form
\[\tilde{r}_{i,j}=\{z \in Z \;|\;I(z,(i,j))=1\}\]
With Theorem~\ref{thm:oracle:l1} we immediately get
\begin{theorem}\label{thm:oracle:l}
One can build a data structure of size $O(m\ell+ m^2)$ in time $O\left(m^2+m n\log(n)\right)$ and  $O(n+m\ell+m^2)$ space that answers for an element of the ground set $Z$ and a set of ${\ASpace}_4$, whether this element is contained in the set in $O(\ell^2)$ time.
\end{theorem}

Since for all $(i,j)$ we have that $r_{i,j}\subseteq \tilde{r}_{i,j}$ it holds that for each set cover of $\ASpace_3$, there is also a set cover of the same size for ${\ASpace}_4$. Together with Lemma~\ref{lem:hittingpoly:2} this directly implies 
\begin{lemma}\label{lem:hittingpoly:2:approxorac}
If there exists a set cover $S$ of $\RSpace$, then there exists a set cover of the same size for ${\ASpace}_4$.
\end{lemma}
For the set system $\ASpace_4$ we further can derive a lemma corresponding to Lemma~\ref{lem:hittingpoly:1} using that for $z\in\tilde{r}_{i,j}$ we have $d_F(
\kappa_z(x,y), \approxcurvesimpl^+(i,j) )  \leq 46\Delta$. The proof is in all other parts completely analogous.
\begin{lemma}\label{lem:hittingpoly:1:approxorac} 
Assume there exists a set cover for $\RSpace$ with parameter $\Delta$.
Let $S$ be a set cover of size $k$ for $\ASpace_4$. We can derive from $S$ a set of $k$ cluster centers $C \subseteq \XX^d_{l}$ and such that $\phi(P,C) \leq 50\Delta$.
\end{lemma}
So if  we apply Theorem~\ref{thm:framework}  to the set system ${\ASpace}_4$ given by the approximation oracle we merely lose a constant approximation factor for our clustering problem in comparison to the direct application on the set system $\ASpace_3$. This leads to the following result.

\subsection{The result}

\begin{lemma}\label{thm:runtimeline:d} 
Let $k$ be the minimum size of a set cover for ${\ASpace}_4$. There exists an algorithm that computes a set cover for ${\ASpace}_4$ of size $O( k\log^2 (m) )$ with expected running time in 
$\widetilde{O}\left( k \ell^2 m^2 +mn \right)$
 and using space in $O(n+m\ell+m^2)$.
\end{lemma}

\begin{proof}
Note that we must have $k<m$ if such a set $C^*$ exists. Indeed, this is the case  since for each $i\in\{1,\dots,m-1\}$ the subcurve $P[t_i,t_{i+1}]$ has to be covered by only one element of $C^*$. So if we had $k>m-1$ then we would have more center curves in $C^*$ than elements to cover. 
We apply Theorem~\ref{thm:framework} to compute a set cover of $(Z,{\ASpace}_4)$. For Theorem~\ref{thm:framework}, we use Theorem~\ref{thm:oracle:l}, $|Z|= m-1$ and $|{\ASpace}_4|=O(m^2)$. Again, the VC-dimension of the dual set system is bounded by $O(\log m)$.
\end{proof}

\begin{restatable}{theorem}{mainTheorem}\label{thm:main}
Let $\curveP: [0,1] \rightarrow \RR^d$ be a polygonal curve of complexity $n$ with breakpoints $0 \leq t_1,\dots,t_m\leq 1$. Assume there exists a set $C^{*}\subset \XX^d_{\ell}$ of size $k \leq m$, such that $\phi(P,C^{*})\leq \Delta$. Then there exists an algorithm that computes a set $C \subset \XX^d_{\ell}$ of size $O( k\log (m) \log^2 (m))$ such that $\phi(P,C)\leq 50\Delta$. The algorithm has  expected running time in
$\tilde{O}\left(
k \ell^2 m^2 + mn \right)$
and uses space in $O(n+m\ell+m^2)$.
\end{restatable}

\begin{proof}
The theorem follows directly by the combination of Lemma~\ref{thm:runtimeline:d}, Lemma~\ref{lem:hittingpoly:2:approxorac}  and Lemma~\ref{lem:hittingpoly:1:approxorac}.
\end{proof}

\section{Improving the algorithm in  the continuous case}\label{sec:cont}

In the previous sections we considered  the discrete variant of the subtrajectory clustering problem, assuming we are given breakpoints that denote the possible start and end points of subcurves that cover $P$. In the continuous case, we do not restrict the subcurves of $P$ to start and end at breakpoints. Recall that a point of $P$ is covered by a center curve $c$ if there is any subcurve $S$ of $P$ that contains $p$ and is in Fréchet distance at most $\Delta$ to $c$. In the continuous case we do not restrict $S$ to start and end at a breakpoint of $P$. The exact problem statement is given in Section~\ref{sec:def}.

In this section, we present an approximation algorithm that applies the algorithmic ideas developed in the previous sections to the discretization described in Section~\ref{sec:setup:setsystem}. A direct application of Theorem~\ref{thm:main} using Lemma~\ref{lem:main:approx}, however, leads to a high dependency on the arclength of the input curve, see also the discussion in Section~\ref{sec:overview}. We will see that some steps of the algorithm can be simplified for this particular choice of breakpoints, ultimately leading to an improvement in the running time. Again, the crucial step is to choose the set system and the set system oracle wisely.

\subsection{The set system}\label{sec:cont:approx}

We will again use the set system $\ASpace_3$ that was defined in Section~\ref{sec:setsystemgeneral}. Here we choose $m=\lceil \frac{L}{\epsilon \Delta}\rceil$ breakpoints to ensure that two consecutive breakpoints have a distance of at most $\epsilon \Delta$. The explicit choice of breakpoints was already described in Section~\ref{sec:setup:setsystem}. For the construction of the approximation oracle we then can  take advantage of the fact that two consecutive breakpoints are close to each other. This will allow us to achieve better running time results based on the simpler structure of the oracle. A key factor here is the low VC-dimension of the set system that is dual to the set system which is implicitly given by the oracle.

\subsection{The approximation oracle}\label{sec:cont:oracle}
The new approximation oracle will have the following properties. Given a set $r_{i,j} \in \ASpace_3$
and an element $z \in Z$ this approximation oracle returns either one of the answers below:
\begin{compactenum}[(i)]
\item "Yes", in this case there exists $x \in [x_z,z]$ and $ y \in [z+1,y_{z+1}]$ with $d_F(
P[t_x,t_y], \approxcurvesimpl^+(i,j) )  \leq (14+\epsilon)\Delta$ 
\item  "No", in this case $(i,j)\notin r_z$.
\end{compactenum}
In both cases the answer is correct. Furthermore, we say that the new approximation oracle  answers the query in the same way as the approximation oracle introduced in section~\ref{sec:approxorac} and therefore also needs the same data structures as before. There is only one exception.
The oracle does not need to check if any edge is active and only needs to check if there is a monotone increasing path in the $10\Delta$-free space of $\approxcurvesimpl^+(i,j)$ and $\kappa_z(x_z,y_{z+1})$ that starts before or at $z$ and ends after or at $z+1$. So it also does not need to build the data structure for determining active edges. Neither does it have to save the first and last breakpoint on the edge of each simplification. As a direct consequence we get the following running time result for the new approximation oracle.

\begin{theorem}\label{thm:oracle:cont}
One can build a data structure for the approximation oracle of size $O(m\ell)$ in time $O\left(m n\log(n)\right)$ and  space $O(n+m\ell)$ that has a query time of $O(\ell^2)$. 
\end{theorem}
\paragraph{Correctness.} We want to show that the oracle is still correct, even though it does not check for active edges. To do so, we proof the following lemma.
\begin{lemma}\label{lem:oracle:cont:correct}
Let $z,i,j\in\{1,\dots,m\}$. Consider the query $z\in r_{i,j}$. If the approximation oracle returns the answer 
\begin{compactenum}[(i)]
\item "Yes", then there exists $x \in [x_z,z]$ and $ y \in [z+1,y_{z+1}]$ with $d_F(
P[t_x,t_y], \approxcurvesimpl^+(i,j) )  \leq (14+\epsilon)\Delta$
\item  "No", then we have  $z \notin r_{i,j}$.
\end{compactenum}
\label{lem:correctaprorcont}
\end{lemma}

\begin{proof}
(i) If the oracle returns "Yes", then there exists a monotone increasing path in the $10\Delta$-free space of $\approxcurvesimpl^+(i,j)$ and $\kappa_z(x_z,y_{z+1})$ that starts before or at $z$ and ends after or at $z+1$. Let $p$ be the start of the path on $\approxcurvesimpl^-(x_z,z)$. Let $q$ be a point of $P$ that gets mapped to $p$ by a strictly monotone increasing function from $P[t_{x_z},t_z]$ to $\approxcurvesimpl^-(x_z,z)$ that realises the Fréchet distance $d_F(P[t_{x_z},t_z],\approxcurvesimpl^-(x_z,z))$. So the last breakpoint $q_{\epsilon}$ before $q$ has distance at most $\epsilon\Delta$ to $q$. Therefore we have by triangle inequality
\[d(p,q_\epsilon)\leq d(p,q)+d(q,q_\epsilon)\leq (4+ \epsilon)\Delta\]
Since $d(p,q_\epsilon)\leq (4+ \epsilon)\Delta$ and $d(p,q)\leq 4\Delta$, we also have for the line segment $\overline{q_\epsilon q} $ that
\[d_F(p,\overline{q_\epsilon q})\leq (4+ \epsilon)\Delta\]
An analogous argument can be made for the end point $v$ of the path. So let $v$ get mapped to a point $u$ on $P$ by a strictly monotone increasing function from $\approxcurvesimpl^-(x_z,z)$ to $P[t_{x_z},t_z]$ that realises the Fréchet distance $d_F(P[t_{x_z},t_z],\approxcurvesimpl^-(x_z,z))$.  For the first breakpoint $u_\eps$ after $u$, we therefore get
\[d_F(v,\overline{u u_\epsilon})\leq (4+ \epsilon)\Delta\]
Let $\widetilde{\kappa}$ be the subcurve of $\kappa_z(x_z,y_{z+1})$ starting at $p$ and ending at $v$ and $\widetilde{P}$ be the subcurve of $P$ starting at $q$ and ending at $u$. By the definition of $\kappa_z(x_z,y_{z+1})$ as a $(4\Delta,2\ell)$-simplification and the choices of $p,q,u$ and $v$, we get
\[d_F(\widetilde{\kappa},\widetilde{P})\leq 4\Delta\]
So by concetation we can get the curve
\[\widetilde{P}_\epsilon=\overline{q_\epsilon q}\oplus\widetilde{P}\oplus\overline{u u_\epsilon}\]
which is a subcurve of $P$ with
\[d_F(\widetilde{\kappa},\widetilde{P}_\epsilon)\leq (4+ \epsilon)\Delta\]
By the use of triangle inequality, we now get 
\[d_F(\approxcurvesimpl^+(i,j),\widetilde{P}_\epsilon)\leq d_F(\approxcurvesimpl^+(i,j),\widetilde{\kappa})+d_F(\widetilde{\kappa},\widetilde{P}_\epsilon)\leq (14+\epsilon)\Delta\]
(ii) We prove that the oracle returns the answer "Yes" if $z \in r_{i,j}$:

So let $z\in r_{i,j}$ Then we have $d_F(
\kappa_z(x,y), \approxcurvesimpl^+(i,j) )  \leq 10\Delta$ for some $x_z\leq x\leq z$ and $z+1\leq y\leq y_{z+1}$. Therefore there is a path in the free space diagram that starts before or at $z$ and ends after or at $z+1$.
\end{proof}

Now that we have shown that the oracle works correctly, we describe how we can use the oracle to approximate our problem.
Analogous to the approach in the discrete case, we define a set system that is implicitly given by the new approximation oracle. 
Let $\widetilde{I}(z,(i,j))$ be the output of the approximation oracle for $z,i,j\in\{1,\dots,m\}$ with
\begin{align*}
    \widetilde{I}(z,(i,j))&=1 \;\;\;\;\text{ if the oracle answers "Yes"}\\
    \widetilde{I}(z,(i,j))&=0 \;\;\;\;\text{ if the oracle answers "No" } 
\end{align*}
Let ${\ASpace}_{5}$ be the set system consisting of sets of the form
\[\tilde{r}_{i,j}=\{z \in Z \;|\;\widetilde{I}(z,(i,j))=1\}\]
With Theorem~\ref{thm:oracle:cont} we immediately get
\begin{theorem}\label{thm:oracle:cont:2}
One can build a data structure of size $O(m\ell)$ in time $O\left(m n\log(n)\right)$ and  $O(n+m\ell)$ space that answers for a breakpoint $z\in\{1,\dots,m\}$ and a set of ${\ASpace}_5$, whether $z$ is contained in the set in $O(\ell^2)$ time.
\end{theorem}
We can also get the following  results for the set system $\ASpace_5$ in the same way as before.
We use that each range in $\ASpace_3$ is contained in a range of $\ASpace_5$. Together with Lemma~\ref{lem:hittingpoly:2} this directly implies 
\begin{lemma}\label{lem:hittingpoly:2:cont}
If there exists a set cover $S$ of $\RSpace$, then there exists a set cover of the same size for ${\ASpace}_5$.
\end{lemma}
To get the next result, we use that for $z\in\tilde{r}_{i,j}$ we have $d_F(
\kappa_z(x,y), \approxcurvesimpl^+(i,j) )  \leq (14+\epsilon)\Delta$. Imitating the proof of Lemma ~\ref{lem:hittingpoly:1} we then get 
\begin{lemma}\label{lem:hittingpoly:1:cont} 
Assume there exists a set cover for $\RSpace$ with parameter $\Delta$.
Let $S$ be a set cover of size $k$ for $\ASpace_4$. We can derive from $S$ a set of $3k$ cluster centers $C \subseteq \XX^d_{l}$ and such that $\phi(P,C) \leq (18+\epsilon)\Delta$.
\end{lemma}

These results imply that a minium set cover of ${\ASpace}_{5}$ can be used to find an approximate solution for our clustering problem. But to apply Theorem~\ref{thm:frameworkhittingset} for finding a good set cover, we first need to bound the VC-dimension of the dual of ${\ASpace}_{5}$.

\subsection{The VC-dimension}\label{sec:cont:vcdim}

To bound the VC-dimension of the set system ${\ASpace}_{5}$ and its dual set system, we use the methods introduced by Driemel et al. \cite{driemel2019vc}. Leading up to that, we first show analogous to the proof of Lemma 9 in \cite{arxAD17}  that the output $\widetilde{I}(z,(i,j))$ of the approximation oracle  can be determined by the truth value of the following predicates $\pazocal{P}_1,\pazocal{P}_2,\pazocal{P}_3,\pazocal{P}_4$ for  $\approxcurvesimpl^+(i,j)$ and $\kappa_z(x_z,y_{z+1})$. Let $s_1,\dots,s_{\ell_1}$ be the vertices of $\approxcurvesimpl^+(i,j)$ and $q_1,\dots,q_{\ell_2}$ be the vertices of $\kappa_z(x_z,y_{z+1})$. Note that we have $\ell_1=O(\ell)$ and $\ell_2=O(\ell)$. We define

\begin{compactenum}[$\pazocal{P}_1$]
\item  (Vertex-edge (vertical)) Given an edge of $\approxcurvesimpl^+(i,j), \overline{s_j s_{j+1}}$ and a vertex $q_i$ of $\kappa_z(x_z,y_{z+1})$, this predicate returns true iff there exists a point $p\in \overline{s_j s_{j+1}}$, such that $\|p-q_i\|\leq 10\Delta$.

\item (Vertex-edge (horizontal)) Given an edge of $\kappa_z(x_z,y_{z+1}), \overline{q_i q_{i+1}}$ and a vertex $s_j$ of $\approxcurvesimpl^+(i,j)$, this predicate returns true iff there exists a point $p\in \overline{q_i q_{i+1}}$, such that $\|p-s_j\|\leq 10\Delta$.

\item (Monotonicity (vertical)) Given two vertices of $\approxcurvesimpl^+(i,j)$, $s_j$ and $s_t$ with $j<t$ and an edge of $\kappa_z(x_z,y_{z+1})$, $\overline{q_i q_{i+1}}$, this predicate returns true if there exists two points $p_1$ and $p_2$ on the line supporting the directed edge, such that $p_1$ appears before $p_2$ on this line, and such that $\|p_1-s_t\|\leq 10\Delta$ and $\|p_2-s_j\|\leq 10\Delta$.

\item (Monotonicity (horizontal)) Given two vetices of $\kappa_z(x_z,y_{z+1})$, $q_i$ and $q_t$ with $i<t$ and an edge of $\approxcurvesimpl^+(i,j)$, $\overline{s_j s_{j+1}}$, this predicate returns true if there exists two points $p_1$ and $p_2$ on the line supporting the directed edge, such that $p_1$ appears before $p_2$ on this line, and such that $\|p_1-q_t\|\leq 10\Delta$ and $\|p_2-q_i\|\leq 10\Delta$.\\
\end{compactenum}

To show our claim we use the following lemma.
\begin{lemma}\label{lem:vcdim:cont:analog}
Let $s$ and $q$ be two polygonal curves with vertices $s_1,\dots,s_{\ell_1}$ and  $q_1,\dots,q_{\ell_2}$. Let further $1\leq a\leq b\leq \ell_2$. Given the truth value of all predicates $\pazocal{P}_1,\dots,\pazocal{P}_4$, one can determine if there exists a monotone increasing path in the $10\Delta$-free space of $s$ and $q$ that starts in $\overline{q_a q_{a+1}}$ at the bottom of the free space diagram and ends in $\overline{q_b q_{b+1}}$ at the top of the free space diagram. 
\end{lemma}
The proof of Lemma~\ref{lem:vcdim:cont:analog} is analogous to the proof of Lemma 9 in \cite{arxAD17} and much of the argumentation can be applied verbatim. We include the proof here for the sake of completeness, since there are some subtle differences.

\begin{proof}
As in the proof of Lemma 9 in \cite{arxAD17}, we first introduce the notion of valid sequence of cells in the free space diagram. We as well denote the cell corresponding to the edges $\overline{q_i q_{i+1}}$ and $\overline{s_j s_{j+1}}$ with $C_{i,j}$. The definition of a valid sequence, however, changes slightly for our application. We call a sequence of cells $\mathcal{C}=((i_1,j_1),(i_2,j_2),\dots,(i_k,j_k))$ valid if $i_1=a,j_1=1,i_k=b,j_k=\ell_1-1$ and if for any two consecutive cells $(i_m,j_m)$ and $(i_{m+1},j_{m+1})$ it holds that either $i_m=i_{m+1}$ and $j_{m+1}=j_m +1$ or $j_m=j_{m+1}$ and $i_{m+1}=i_m +1$. The only difference to the definition in~\cite{arxAD17} is that we require $i_1=a$ and $i_k=b$. 

In our application we say that a monotone increasing path in the $10\Delta$-free space of $s$ and $q$ is feasible if it starts in $\overline{q_a q_{a+1}}$ at the bottom of the free space diagram and ends in $\overline{q_b q_{b+1}}$ at the top of the free space diagram. It is easy to see that for any valid sequence there exists a feasible path which passes the cells in the order of the sequence. On the other hand, it is also true that for each feasible path there exists a valid sequence such that the path passes the cells in the order of the sequence. In the following, we identify with each sequence of cells $\mathcal{C}$ a set of predicates $\mathcal{P}$. The set of predicates is different from the predicates in~\cite{arxAD17} and consist of the following.
\begin{compactenum}[i)]

\item $(\pazocal{P}_1)_{(i,j)} \in \mathcal{P}$ iff $(i,j-1),(i,j)\in \mathcal{C}$.
\item $(\pazocal{P}_2)_{(i,j)} \in \mathcal{P}$ iff $(i-1,j),(i,j)\in \mathcal{C}$.
\item  $(\pazocal{P}_2)_{(a,1)} \in \mathcal{P}$ and $(\pazocal{P}_2)_{(b,\ell_1)} \in \mathcal{P}$
\item $(\pazocal{P}_3)_{(i,j,k)} \in \mathcal{P}$ iff $(i,j-1),(i,k)\in \mathcal{C}$ and $j<k$.
\item $(\pazocal{P}_3)_{(a,1,k)} \in \mathcal{P}$ iff $(a,k)\in \mathcal{C}$ and $1<k$.
\item $(\pazocal{P}_3)_{(b,j,\ell_1-1)} \in \mathcal{P}$ iff $(b,j)\in \mathcal{C}$ and $j<\ell_1-1$.
\item $(\pazocal{P}_4)_{(i,j,k)} \in \mathcal{P}$ iff $(i-1,j),(k,j)\in \mathcal{C}$ and $i<k$.
\end{compactenum}
As in~\cite{arxAD17}, we say that a valid sequence of cells is feasible if the conjunction of its induced predicates is
true. We claim that any feasible path through the free-space induces a feasible sequence of cells and vice versa. To prove the claim we use the following helper lemma from~\cite{arxAD17}. 

\begin{lemma}[\cite{arxAD17}, Lemma 10]\label{lem:helpvccont}
Let $\mathcal{C}$ be a feasible sequence of cells and consider a monotonicity predicate $\pazocal{P}$ of the set of predicates $\mathcal{P}$ induced by $\mathcal{C}$. Let $a_1$ and $a_2$ be the vertices and let $e$ be the directed edge associated with $\pazocal{P}$. There exist two points $p_1$ and $p_2$ on $e$, such that $p_1$ appears before $p_2$ on $e$, and such that $\|p_1-a_1\|\leq 10\Delta$ and $\|p_2-a_2\|\leq 10\Delta$. 
\end{lemma}

Lemma~\ref{lem:helpvccont} holds for our definition of feasible sequences of cells in the same way as in the original work. For the proof, we refer to \cite{arxAD17}. To continue the proof of Lemma~\ref{lem:vcdim:cont:analog}, we claim that any feasible path induces a feasible sequence of cells and vice versa. Assume there exists a feasible path $\pi$ that passes through the sequence of cells $\mathcal{C}$. The truth value of the predicates $(\pazocal{P}_2)_{(a,1)}$ and $(\pazocal{P}_2)_{(b,\ell_1)}$ follows directly by the starting and ending conditions of a feasible path. The truth value of the other predicates can be derived in the following way (which is exactly the same as in \cite{arxAD17}).

Consider a horizontal vertex-edge predicate $(\pazocal{P}_2)_{(i,j)}$ for consecutive pairs of cells $C_{(i,j-1)}$, $C_{(i,j)}$ in the sequence $\mathcal{C}$. The path $\pi$ is a feasible path that passed through the cell boundary between these two cells. This implies that the there exists a point on the edge  $\overline{q_iq_{i+1}}$ which lies within distance $10\Delta$ to the vertex $s_j$. This implies that the predicate is true. A similar argument can be made for each vertex-edge predicate.

Next, we will discuss the monotonicity predicates. Consider a subsequence of cells of $\mathcal{C}$ that lies in a fixed column $i$ and consider the set of predicates  $\mathcal{P}'\subseteq \mathcal{P}$ that consists of vertical monotonicity
predicates $(\pazocal{P}_3)_{(i,j,k)}$ for fixed $i$. Let $p_j,p_{j+1},\dots,p_k$ be the sequence of points along $q$ that correspond to the vertical coordinates where the path $\pi$ passes through the corresponding cell boundaries
corresponding to vertices  $s_j,s_{j+1},\dots,s_k$. The sequence of points lies on the directed line supporting the edge $\overline{q_i q_{i+1}}$ and the points appear in their order along this line in the sequence due to the monotonicity of $\pi$. Since $\pi$ is a feasible path it lies in the free-space and therefore we have $\|p_{k'}-s_{k'}\|\leq 10\Delta$ for every $j\leq k' \leq k$. This implies that all predicates in $\mathcal{P}$ are true. We can make a similar argument for the horizontal monotonicity predicates $(\pazocal{P}_4)_{(i,j,k)}$ for a fixed row $j$. This shows that a feasible path $\pi$ that passes through the cells of $\mathcal{C}$ implies that the conjunction of induced predicates $\mathcal{P}$ is true.

It remains to show the other direction. Since each cell of the free space is convex, it is clear that the vertex edge predicates give us the existence of a continuous (not necessarily monotone) path $\pi$ that stays inside the free space and connects the edges $\overline{q_a q_{a+1}}$ and $\overline{q_b q_{b+1}}$. To show that there always exists such a path that is also $(x,y)$-monotone we again use the argumentation of \cite{arxAD17}.

Assume for the sake of contradiction that the conjunction of predicates in $\mathcal{P}$ is true, but there
exists no feasible path through the sequence of cells $\mathcal{C}$ . In this case, it must be that either a
horizontal passage or a vertical passage is not possible. Concretely, in the first case, there must be
two vertices $s_j$ and $s_k$ and a directed edge $e = \overline{q_i q_{i+1}}$, such that there exist no two points $p_1$ and $p_2$
on $e$, such that $p_1$ appears before $p_2$ on $e$, and such that $\|p_1 -s_j\| \leq 10\Delta $ and $\|p_2 - s_k\| \leq 10\Delta$. However,
$(P_3)_{(i,j,k)}$ is contained in $\mathcal{P}$ and by Lemma~\ref{lem:helpvccont} two such points $p_1$ and $p_2$ must exist. We obtain a
contradiction. In the second case, the argument is similar. Therefore, a feasible sequences of cells
implies a feasible path, as claimed.
\end{proof}
Lemma~\ref{lem:vcdim:cont:analog} now directly implies the following theorem.
\begin{theorem}\label{thm:pred:orac}
  Given the truth values of all predicates $\pazocal{P}_1,\dots, \pazocal{P}_4$ for two fixed curves $\approxcurvesimpl^+(i,j)$ and $\kappa_z(x_z,y_{z+1})$, one can determine the value of $\widetilde{I}(z,(i,j))$.
\end{theorem}

We use Theorem~\ref{thm:pred:orac} to determine the following bound on the VC-dimension of $(Z,\ASpace_5)$.

\begin{theorem}\label{thm:vc:cont} Let $Z=\{1,\dots,m\}$.
The VC-dimension of $(Z,\ASpace_5)$ and its dual set system are both in $O(d^2\ell^2\log(d\ell))$.
\end{theorem}

The proof of Theorem~\ref{thm:vc:cont} is analogous to the proof of Theorem 27 in \cite{driemel2019vc} and included here for the sake of completeness. For the proof we use VC-dimension bounds for the following set systems.
\begin{definition}
For any two points $s,t\in \RR^d$ and $r\in\RR_+$ define the stadium centered at $\overline{st}$ as
\[D_r(\overline{st})=\{x\in \RR^d \,\vert\, \exists p\in \overline{st},\|p-x\|\leq r\}\]
We further define the monotony sets $M_r(\overline{st})\subseteq \XX^d_2$ as the sets where $\{q_1,q_2\}\in M_r(\overline{st})$ if and only if
\begin{itemize}
    \item $\exists p_1, p_2\in\ell$ where $\overline{st}$ supports $\ell$ such that:
    \item $\|q_1-q_2\|\leq r$ and $\|p_2-q_2\|\leq r$; and
    \item $p_1$ is less than $p_2$ along the line as $\langle p_1,t-s\rangle\leq \langle p_2,t-s\rangle$
\end{itemize}
The resulting set systems are then $(\RR^d,\mathcal{D})$ with $\mathcal{D}=\{D_r(\overline{st})\,\vert\, s,t\in \RR^d, r\in\RR_+\}$ and $(\XX^d_2,\mathcal{M})$ with $\mathcal{M}=\{M_r(\overline{st})\,\vert\, s,t\in \RR^d, r\in\RR_+\}$.
\end{definition}

\begin{theorem}[\cite{driemel2019vc}, Corollary 15]\label{thm:vcdim:stad}
The VC-dimension of $(\RR^d,\mathcal{D})$ and its dual set system are both in $O(d^2)$.
\end{theorem}
\begin{theorem}[\cite{driemel2019vc}, Corollary 26]\label{thm:vcdim:mon}
The VC-dimension of $(\XX^d_2,\mathcal{M})$ and its dual set system are both in $O(d^2)$.
\end{theorem}

With the help of these bounds we can now proof Theorem~\ref{thm:vc:cont}.

\noindent\textit{Proof of Theorem~\ref{thm:vc:cont}.}
Let $S\subseteq \{1,\dots,m\}$ be a set of $t$ breakpoints and let $(i,j)\in\{1,\dots,m\}^2$. Due to 
Theorem~\ref{thm:pred:orac} we have that the set $\{z\in S \,\vert\, \widetilde{I}(z,(i,j))=1\}$ is uniquely defined by the sets
\[
    \bigcup_{k=1}^4\bigcup_{z\in S}\pazocal{P}_k^{10\Delta}(\approxcurvesimpl^+(i,j),\kappa_z(x_z,y_{z+1}))
\]
Using Theorem~\ref{thm:vcdim:stad}, we can bound the number of all possible sets $\bigcup_{z\in S}\pazocal{P}_1^{10\Delta}(\approxcurvesimpl^+(i,j),\kappa_z(x_z,y_{z+1}))$ and the number of all possible sets $\bigcup_{z\in S}\pazocal{P}_2^{10\Delta}(\approxcurvesimpl^+(i,j),\kappa_z(x_z,y_{z+1}))$ both by $(t\ell)^{O(d^2\ell)}$.
Furthermore, the number of all possible sets $\bigcup_{z\in S}\pazocal{P}_3^{10\Delta}(\approxcurvesimpl^+(i,j),\kappa_z(x_z,y_{z+1}))$ and the number of all possible sets $\bigcup_{z\in S}\pazocal{P}_4^{10\Delta}(\approxcurvesimpl^+(i,j),\kappa_z(x_z,y_{z+1}))$ are both bounded by $(t\ell)^{O(d^2\ell^2)}$ by Theorem~\ref{thm:vcdim:mon}. The $\ell^2$ term arises beause we consider $\Theta(\ell^2)$ pairs $s_j,s_t$ for Predicate $\pazocal{P}_3$ ($q_i,q_t$ for Predicate $\pazocal{P}_4$). Hence, we get
\[2^t\leq (t\ell)^{O(d^2\ell^2)}\implies t=O(d^2\ell^2 \log(d\ell)).\]
\qed

\subsection{The result}\label{sec:cont:result}

We apply Theorem~\ref{thm:frameworkhittingset} on the dual of  $(\{1,\dots,m\},\ASpace_5)$ to get the following result for computing a set cover. We use here that $|{\ASpace}_5|=O(m^2)$ and that the resut of Theorem~\ref{thm:vc:cont} that the VC-dimension of ${\ASpace}_5^*$ is in $O(d^2\ell^2\log(d\ell))$.

\begin{lemma}\label{thm:runtime:cont} 
Let $k$ be the minimum size of a set cover for ${\ASpace}_5$. Let further $m=\lceil\frac{L}{\epsilon\Delta}\rceil$  and $\delta= O(d^2\ell^2\log(d\ell)))$,   there exists an algorithm that computes a set cover for ${\ASpace}_5$ of size $O( k\delta \log (\delta k))$ with  expected running time in 
$\widetilde{O}(  k \ell^2  \delta m^2 + mn ) $
 and using space in $O(n+m\ell)$.
\end{lemma}

This lemma finally implies our main result for the clustering problem in the continuous case.

\mainTheoremcont*

\begin{proof}
The theorem follows immediately by the combination of Theorem~\ref{thm:runtime:cont}, Lemma~\ref{lem:hittingpoly:2:cont}  and Lemma~\ref{lem:hittingpoly:1:cont}.
\end{proof}

\section{Additional lower bounds for the VC-dimension}\label{sec:vcdim}

In this section we derive bounds on the VC-dimension of the dual set systems in the discrete and continuous case. 
Consider the set system $\RSpace$ from Section~\ref{sec:setup:setsystem}. The dual set system of $\RSpace$ is the set system $\RSpace^*$ with ground set $\XX^d_\ell$ where each set $r_{z} \in
\RSpace^*$ is defined by a breakpoint $z \in \{1,\dots,m-1\}$ as follows
\[ 
r_{z} = \left\{ Q \in \XX^d_\ell \mid \exists i \leq z < j  \text{ with }  d_F(Q,\curveP[t_i,t_j]) \leq \Delta \right\}
\]
In the continuous case, the dual set system is the set system $\RSpace^{*}_0$ with ground set $\XX^d_{\ell}$ where each set $r_{t} \in
\RSpace^{*}_0$ is defined by a parameter $t\in[0,1]$  as follows
\[
r_{t} = \left\{ Q \in \XX^d_\ell \mid \exists  t' \leq t < t''  \text{ with }  d_F(Q,\curveP[t',t'']) \leq \Delta \right\}
\]

Before deriving bounds on the VC-dimension of $\RSpace^*$ and $\RSpace^*_0$ in the general case, we observe that in the special case where cluster centers are points $\ell=1$, there is a simple upper bound to the VC-dimension. In this chapter, we use the notation 
$b(p,\rho) = \{ q \in \RR^d \mid \|p-q\| \leq \rho \}$ for the Euclidean ball of radius $\rho\geq0$ centered at $p\in \RR^d$.
\begin{lemma}
For $\ell=1$, the VC-dimension of $(\XX^d_\ell,\RSpace^*_0)$ and $(\XX^d_\ell,\RSpace^*)$  are both at most $d+1$.
\end{lemma}
\begin{proof}
We prove the bound for $(\XX^d_\ell,\RSpace^*_0)$ here. The proof works verbatim for $(\XX^d_\ell,\RSpace^*)$.
The ground set of the set system is $\XX^d_{\ell}=\RR^d$. Now, consider a fixed $t \in [0,1]$ and radius $\Delta>0$. We claim that \[r_{t} = b(P(t),\Delta). \]
Indeed, for any $0 \leq c \leq t \leq d \leq 1$, we can write for the set 
\[ R_{[c,d]} = \{ p \in \RR^d \mid d(p,P[c,d]) \leq \Delta \} = \bigcap_{t \in [c,d]} \{ p \in \RR^d \mid \|p - P(t)\| \leq \Delta \} \subseteq b(P(t),\Delta). \] 
Thus, by the definition of $\RSpace^*$,
\[ r_{t} = \bigcup_{0 \leq c \leq t \leq d\leq 1} R_{[c,d]} = b(P(t),\Delta). \]
The claim now follows since the VC-dimension of Euclidean balls in $\RR^d$ is equal to $d+1$. 
\end{proof}

\subsection{Continuous case}\label{sec:vcdim:cont}

We derive a lower bound on the VC-dimension of the dual set system $(\XX^d_\ell,\RSpace^*_0)$ in the general case.

\begin{theorem}\label{thm:vcdim:cont}
For $\ell\geq 2$ and $d\geq 2$, the VC-dimension of $(\XX^d_\ell,\RSpace^*_0)$ is in $\Omega(\log(n))$.
\end{theorem}
 \begin{proof}
 
 \begin{figure}[t]
    \centering
    \includegraphics[width=0.85\textwidth]{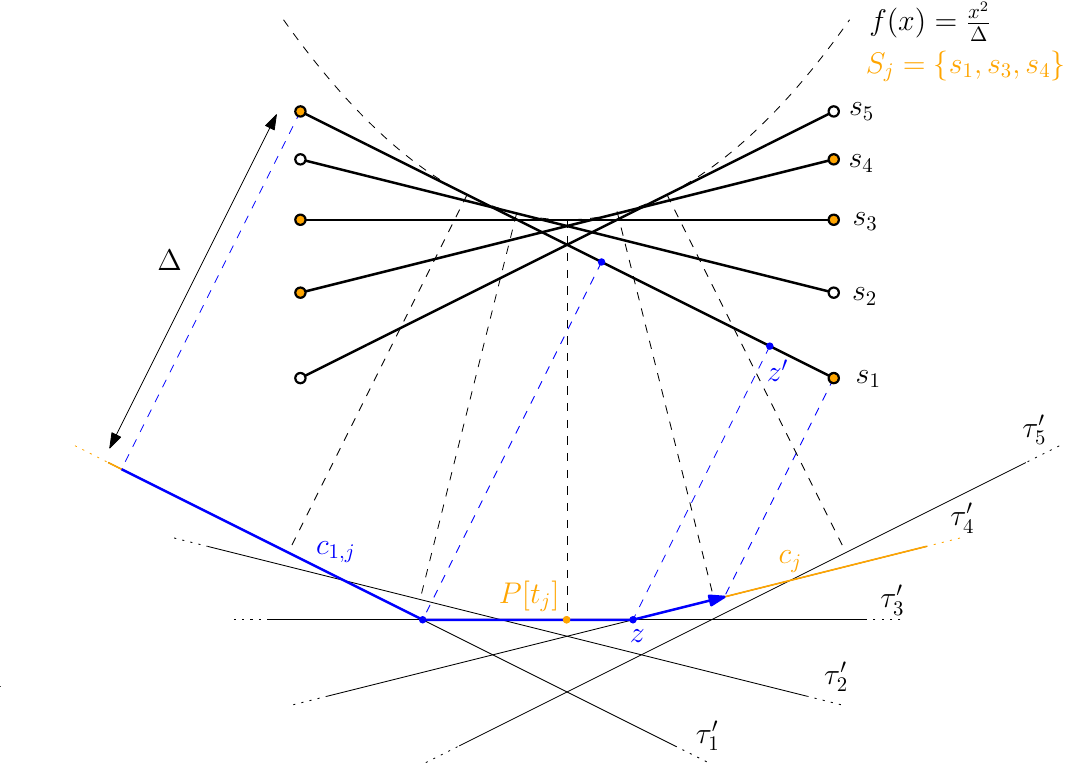}
    \caption{Construction for the case $\ell=2$ such that the VC-dimension of $(\XX^d_\ell,\RSpace^*_0)$ is high. }
    \label{fig:vcdim:lower}
\end{figure}

We show the lower bound for $\ell=2$ and $d=2$; this implies the bound for larger values of $\ell$ and $d$. Let $m\in \NN$. We construct a curve $P$ with at most $O(4^m)$ vertices such that the set system $(\XX^2_{2},\RSpace^*_{0})$ defined on $P$ shatters a set $S\subset \XX^2_2$ of $m$ line segments. 
For the construction of $S=\{s_1,\dots,s_m\}$ we choose line segments that are tangent to the parabola $f(x)=\frac{x^2}{\Delta}$. More specifically, let $\tau_i$ be the tangent that passes through $(x_i,y_i)=(\frac{\Delta(i-1)}{2(m-1)}-\frac{\Delta}{4},\frac{(\frac{\Delta(i-1)}{2(m-1)}-\frac{\Delta}{4})^2}{\Delta})$. Then $s_i$ is the intersection of $\tau_i$ with the rectangle $[-\frac{2\Delta}{3},\frac{2\Delta}{3}]\times [-\frac{\Delta}{2},\frac{\Delta}{2}]$.  The construction is visualized in Figure~\ref{fig:vcdim:lower}. 

 Consider the power set $2^{S}$.  We show that for each subset $S_j \in 2^A$ there exists a curve $c_j \in \XX^d_{m+1}$ such that for each $s_i\in S_j$ there exists a subcurve $c_{i,j}$ of $c_j$ with  $d_F(s_i,c_{i,j})\leq \Delta$ and for each $s_i\in S\setminus S_j$ there exists no subcurve $c_{i,j}$ of $c_j$ with  $d_F(s_i,c_{i,j})\leq \Delta$. The curve $P$ defining the set system instance will later be defined as a concatenation of these curves $c_j$. This will allow us to find a point $t_j$ on $c_j$ such that $r_{t_j}\cap S = S_j$ for each $j$, which then implies that $S$ can be shattered.

The curve $c_j$ can be generated as follows. Let $\tau_i'$ be the line parallel to $\tau_i$ that lies below $\tau_i$ and has distance $\Delta$ to $\tau_i$. For $S_j\in 2^S$ we define with $o_j$ the upper contour set of the lines $\tau_i'$ such that $s_i\in S_j$. We further define $c_j$ to be the intersection of $o_j$ with $[-2\Delta,2\Delta]\times(-\infty,\infty)$.
We observe that for $s_i\in S\setminus S_j$  the intersection of $b((x_i,y_i),\Delta)$ and $c_j$ is empty. Therefore there exists no subcurve $c_{i,j}$ of $c_j$ with $d_F(s_i,c_{i,j})\leq \Delta$. For $s_i=\overline{p_iq_i}\in S_j$ let $l_{p_i}$ (resp. $l_{q_i}$) be the line perpendicular to $s_i$ that contains $p_i$ (resp. $q_i$). We define $c_{i,j}$ to be the subcurve of $c_j$ starting at the intersection of $l_{p_i}$ and $c_j$ and ending at the intersection of $l_{q_i}$ and $c_j$. To show that $d_F(c_{i,j},s_i)\leq \Delta$, we divide $s_i$ into edges by projecting each vertex $z$ of $c_{i,j}$ orthogonal onto $s_i$.
Since the slope of each edge of $c_j$ is between $-\frac{1}{2}$ and $\frac{1}{2}$ and also the slope of $s_i$ is between $-\frac{1}{2}$ and $\frac{1}{2}$, the projected vertices appear in the same order on $s_i$ as the corresponding vertices appear on $c_j$. 

So to conclude that $d_F(c_{i,j},s_i)\leq \Delta$, it remains to show that each vertex $z$ of $c_{i,j}$ has distance at most $\Delta$ to its projection $z'$ on $s_i$. This is enough because the Fr\'echet distance of two edges is attained at the distances of the start points or the end points of the edges. So let $z$ be a vertex of $c_{i,j}$. By construction, $c_{i,j}$ is part of the upper contour set $o_j$. We observe that the rectangle $[-\frac{2\Delta}{3},\frac{2\Delta}{3}]\times [-\frac{\Delta}{2},\frac{\Delta}{2}]$ that contains all line segments $S$ lies in the connected component of $\RR^2\setminus o_j$ that does not contain $\tau_i'$. Therefore the ray starting at $z'$ and containing $\overline{z'z}$ hits $z$ before or at the same time as it hits $\tau_i'$. So we have
    \[d(z',z)\leq d(z',\tau_i')=\Delta\]
Note that the intersection $\bigcap_{i: s_i\in S_j}c_{i,j}$ always contains the intersection of $c_j$ with the vertical axis through $(0,0)$. This is the case because the $x$-coordinate of the start point of each curve $c_{i,j}$ is smaller than $0$ and the $x$-coordinate of the end point of each curve $c_{i,j}$ is greater than $0$.

 Let
\[ \curveP =  \bigoplus_{j=1}^{2^m}  c_j. \]
Since each curve $c_j$ has at most $m+1$ vertices, we get that $P$ has at most $n=m2^m=O(4^m)$ vertices and thus $m$ is in $\Omega(\log_4(n))$.

So, it remains to show that the set $S$ is shattered by $(\XX^2_{2},\RSpace^*_{0})$ defined on $\curveP$. Indeed, for any $S_j \in 2^{S}$, let $t_j\in[0,1]$ be the parameter such that $P[t_j]$ is the intersection of $c_j$ with the vertical axis through $(0,0)$. We claim
\[r_{t_j}\cap S = S_j.\]
Since $P[t_j]\in \bigcap_{i: s_i\in S_j}c_{i,j}$, we get by the analysis above that $S_j\subseteq r_{t_j}\cap S$.

On the other hand, for each $s_i\in S\setminus S_j$ there exists no subcurve $c_{i,j}$ of $c_j$ with  $d_F(s_i,c_{i,j})\leq \Delta$. Note that the start points and end points of $c_j$ are by construction more than $\Delta$ away from any point on $s_i$. Therefore $d_F(s_i,Q)\leq \Delta$ for each subcurve $Q$ of $P$ that contains either the start point or the end point of $c_j$. So in total we get that $s_i\in r_{t_j}\cap S$.
\end{proof}

\subsection{Discrete case}\label{sec:vcdim:discrete}

Now we consider the set system $(\XX^d_\ell,\RSpace^*)$ that is dual to the set system $\RSpace$, which was introduced in Section~\ref{sec:setup:setsystem} to discretize our clustering problem through the addition of breakpoints.

We show that the VC-dimension of $(\XX^d_\ell,\RSpace^*)$  is in $\Theta(\log m)$ in the worst case for any reasonable values of $d$ and $\ell$. Interestingly, our bounds on the VC-dimension are independent of $d$ and $n$. In fact, quite surprisingly, they also hold if $\curveP$ is non-polygonal. The upper bound that the VC-dimension of $\RSpace^*$ is at most $\log(m)$ follows directly from the upper bound on the size of the set system. It remains to show the lower bound.

\begin{theorem}\label{thm:true:vcdim2}\label{thm:vcdim:discrete}
For $d \geq 2$ and $\ell \geq 1$ the VC-dimension of $(\XX^d_\ell,\RSpace^*)$  is in $\Omega(\log m)$ in the worst case.
\end{theorem}

\begin{proof}
We show the lower bound for $\ell=1$ and $d=2$; this implies the bound for larger values of $\ell$ and $d$. To show the lower bound, we need to construct a set $A \subseteq \RR^2$ with $|A| =t$ for $t \in \Omega(\log m)$, and a $\curveP$ with breakpoints $t_1,\dots,t_m$, such that $A$ is shattered by $\RSpace^{*}$ as defined by $\curveP$.

We use the lower bound construction of~\cite{driemel2019vc} for the VC-dimension of the set system of metric balls under the Fréchet distance centered at curves of complexity $t$ on the ground set $\RR^2$. According to this result, we can find a set $A$ of $t$ points in $\RR^2$, such that for every subset $A'\subseteq A$ we can find a curve $\curveP_{A'} \in \XX^{d}_{t}$, such that 
\begin{eqnarray}\label{eq:shattering}
A'= A \cap \{ x \in \RR^2| d_F(x,\curveP_{A'}) \leq \Delta \} 
\end{eqnarray}
We will now construct $\curveP$ as the concatenation of these curves with breakpoints at the start and endpoints of these curves, where to concatenate them we linearly interpolate between the endpoints of consecutive curves. 

In order to show correctness of the resulting construction we observe that the definition of the Fréchet distance can be simplified if one of the curves is a point. Let $x \in \RR^2$ and let $\curveP'=\curveP[t_i,t_j]$, then
\begin{eqnarray}
 d_F(x,\curveP') = \max_{t \in [0,1]} (x,\curveP'(t))
\end{eqnarray}
This implies that for the case $\ell=1$ our set system $\RSpace^{*}$ actually has a simpler structure. In particular, any $r_z \in \RSpace^{*}$ defined by an index $z \in Z$ can be rewritten as follows
\begin{eqnarray}
r_z &=& \left\{ x \in \RR^d \mid \exists i \leq z \leq j \text{ with } d_F(x, \curveP[t_i,t_j]) \leq \Delta \right\} \\
&=& \bigcup_{i \leq z < j } \left\{ x \in \RR^d \mid d_F(x, \curveP[t_i,t_{j}]) \leq \Delta \right\}  \\
&=& \left\{ x \in \RR^d \mid d_F(x, \curveP[t_z,t_{z+1}]) \leq \Delta \right\} 
\end{eqnarray}

Thus, with our choice of $\curveP$ and breakpoints $t_1 < \dots < t_m$, we have that for any $A' \subseteq A$ there exists an index $z$ with $1\leq z < m$, such that $A'=A \cap r_z$ holds as required by (\ref{eq:shattering}).
Finally, the number of breakpoints we used is $m=2^{t+1}$ (two breakpoints for each subset of $A$). Therefore, we have $t \geq \log (m) - 1$.
\end{proof}

\begin{figure}
    \centering
    \includegraphics[width=0.7\textwidth]{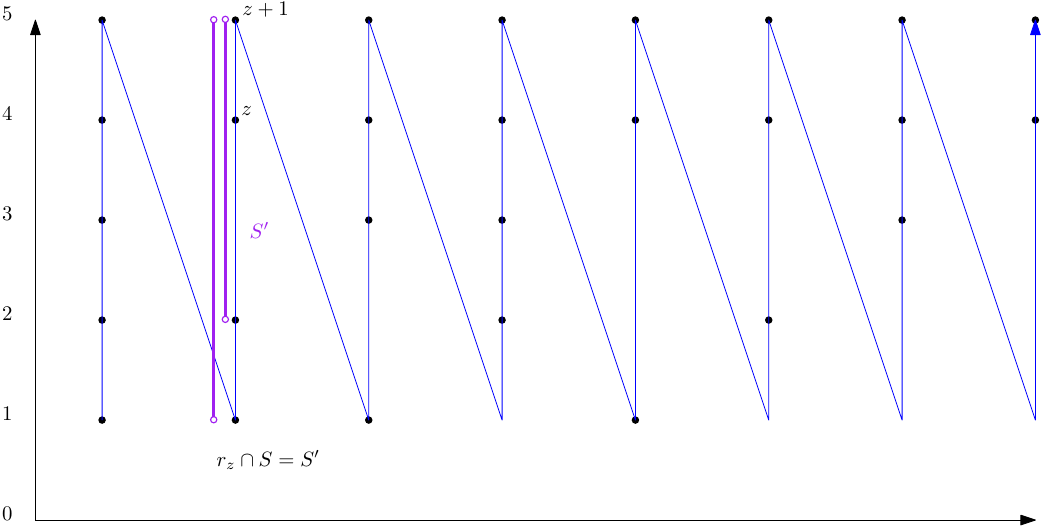}
    \caption{Schematic drawing of $\curveP:[0,1]\rightarrow \RR$ in the construction for the lower bound to the VC-dimension. Parameters of the construction are $\Delta=\frac{1}{3}$, $\ell=2$ and $t=3$. The shattered set of line segments in $\RR$ is $S=\{\overline{1,5}\, \overline{2,5}, \overline{3,5}\}$ with $|S|=t$. The subset encoder segments are shown vertically upwards, the connector segments are shown diagonally downwards. The horizontal axis shows the parametrization of the curve. The figure also shows the subset $S'=\{\overline{1,5}\, \overline{2,5}\}$ and indicates the breakpoint at index $z$, such that $r_z \cap S = S'$.}
    \label{fig:lower_bound_d1}
\end{figure}

\begin{theorem}
For $d\geq 1$ and $\ell \geq 2$ the VC-dimension of $\RSpace^{*}$ is in $\Omega(\log m)$ in the worst case.
\end{theorem}

\begin{proof}
We construct a curve $\curveP$ with breakpoints as follows.
Let $t \in \NN$ be a parameter of the construction. Let $\Delta=\frac{1}{3}$. The curve $\curveP$ is constructed from a series of $2^t$ line segments starting at $0$ and ending at $t+2$ with certain breakpoints along these line segments to be specified later. We call these segments \emph{subset encoder segment}. These line segments are connected by $2^t-1$ line segments starting at $t+2$ and ending at $0$. Those line segments will not contain any breakpoints and we call them \emph{connector segments}. Let $A = \{1,\dots,t\}$ for each subset $A'\subseteq A$ we create one subset encoder segment with breakpoints at the values of $A'$, in addition we put two breakpoints at the values $t+1$ and at $t+2$. 
The curve $\curveP$ is defined by concatenating all $2^t$ subset encoder segments with the connector segments in between. Figure~\ref{fig:lower_bound_d1} shows an example of this construction for $t=3$.
Now, consider the following set of line segments in $\RR$. $S = \{ \overline{s_1 s_2} \mid s_1 \in A, s_2=t+2\}$. We claim that $S$ is shattered by $\RSpace^{*}$ defined on $\curveP$ and $\Delta$. Therefore, the VC-dimension is $t$. The number of breakpoints $m$ we used is upper-bounded by 
$(t+2) 2^{t}$ and therefore $ t \geq \Omega(\log m)$.
\end{proof}

\section{NP-hardness}\label{sec:nphard1}

We note that the problem described in Section~\ref{sec:def} for $\ell=1$ is a specific instance of a $k$-center problem which is NP-complete. 
However, we require that the input set is a connected polygonal curve.
It is tempting to believe that this restriction could make the problem easier.
We show in this section that the problem is still NP-hard.  

Our reduction is from \textsc{Planar-Monotone-3SAT}, with $m_{SAT}$ clauses and $n_{SAT}$ variables.
We show how to construct an instance $B$ of a decision version of our problem given an instance $A$ of \textsc{Planar-Monotone-3SAT}, which is NP-hard~\cite{de2010optimal}. 
We can assume that $A$ is given by a plane rectilinear bipartite graph between variables and clauses where variables are embedded on the x-axis, edges do not cross the x-axis, clauses are adjacent to two or three variables, and are partitioned between  \emph{positive} and \emph{negative} whether they are embedded in the upper or lower half-plane respectively. 
The problem asks whether there exist an assignment from the variables to $\{\texttt{true},\texttt{false}\}$ such that every positive (negative) clause is adjacent to at least one \texttt{true} (\texttt{false}) variable.

\smallskip
\noindent\textbf{Problem definition.}
We define the decision version of our problem as follows.
The instance is defined by a polygonal curve $\curveP$ with breakpoints $0 = t_1 < t_2 < \dots < t_m = 1$ and $\Delta \in \RR$ and $k \in \mathbb{N}$.
The problem asks whether there exist a set $C$ of $k$ points, such that $\phi(P,C) \leq \Delta$. Note that this is equivalent to requiring
\[ \max_{i \in \{1,\dots,m-1\}} \min_{q \in C} \max_{t_i \leq t\leq t_{i+1}} \|\curveP(t)-q\| \leq \Delta \]
A solution $C$ is said to be in \emph{canonical form} if every point in $C$ coincides with a breakpoint, i.e., one of the points $\curveP(t_i)$ for $i \in \{1,\dots,m\}$.

\smallskip
\noindent\textbf{Outline of proof.} We first show how to build an instance $B$ of our problem from $A$. 
We then show that any positive solution $C$ of $B$ can be converted in a positive solution $C'$ in canonical form.
We also show that $C'$ exists if and only if $A$ has a positive solution, which will conclude our proof.
An example of the reduction is shown in Figure~\ref{fig:example}. 
The reduction uses paths formed by unit segments called wires. Figure~\ref{fig:gadgets}~(a) shows circles whose centers represent points in a locally optimal solution.  
Any optimal solution would choose either the red or the blue circles' centers. 
A variable is represented by a cycle as shown in Figure~\ref{fig:gadgets}~(c) formed by $2$ vertical paths and $2$ ``zig-zag'' paths connecting their endpoints. The length of such paths depend on the number of times the variable appears in clauses.
Clauses are represented by a segment whose endpoint is called a clause vertex shown in Figure~\ref{fig:gadgets}~(b) as a star.
It is next to three wires connected to variable gadgets. 
Informally, such segment can be covered by a disk centered at a breakpoint contained in one of the wires if the wire carries a \texttt{true} signal. 

\begin{figure}[h]
    \centering
    \includegraphics[width=.6\textwidth]{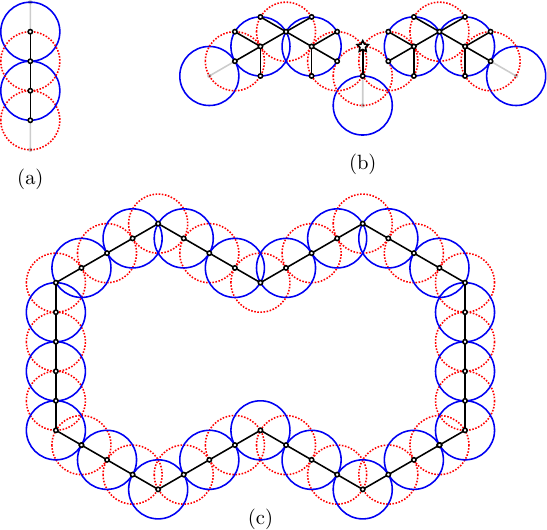}
    \caption{(a) Wire. (b) Clause vertex is shown as a star. (c) Cycle representing a variable with $r=2$.}
    \label{fig:gadgets}
\end{figure}

\smallskip
\noindent\textbf{Construction.}
We modify the embedding in $A$ as follows.
Refer to Figure~\ref{fig:example}.
Replace each variable with a cycle in the hexagonal grid separated by a \emph{separator gadget} shown in Figure~\ref{fig:separator}~(a).
Each cycle contains two vertical edges of length 5 that are all vertically aligned and $4\sqrt{3}r$ apart where $r$ is the maximum number of incidences of the variable in either positive or negative clauses.
In order to close each cycle, connect the upper (resp., lower) endpoints of the vertical edges with a ``zig-zag'' formed by $2r$ edges of length 4 and slopes $1/\sqrt{3}$ and $-1/\sqrt{3}$ (resp., $-1/\sqrt{3}$ and $1/\sqrt{3}$).
We call every even vertex in this upper (resp., lower) ``zig-zag'' path a \emph{positive} (resp., \emph{negative}) \emph{literal vertex}. 
For each clause, the embedding of $A$ allows us to choose two or three literal vertices so that each clause can be connected to literal vertices of their incident variables in a planar way.
For each clause we define three \emph{clause vertices} as follows.
We define positive clauses while negative clauses are defined analogously by reflections.
Let $p_1$, $p_2$, and $p_3$ (if it exists) be the three literal vertices, ordered from left to right, to be connected by the clause, and let $t$ be the smallest distance between them.
The middle clause vertex $c_2$ is above $p_2$ by $t/\sqrt{3}+1$.
Let the left clause vertex $c_1$ (resp., right clause vertex $c_3$)  be $c_2+(-\sqrt{3}/2, -1/2)$ (resp., $c_2+(\sqrt{3}/2, -1/2)$).
Connect $p_2$ to $c_2$ with a vertical edge, and $p_1$ to $c_1$ with a convex with 3 bends as in Figure~\ref{fig:example} so that the length of the vertical edge is 3.
Finally, subdivide each edge into edges of length 1 and at each bend add 6 unit edges as shown in the \emph{turn gadget} in Figure~\ref{fig:separator}~(c).
We obtain an embedding of a graph $G$ containing only unit edges.
We partition the edges of $G$ into two subsets $E_1$ and  $E_2$ as follows. The set $E_1$ contains edges in separator gadgets, the $6$ added edges in each turn gadget and the edge adjacent to $c_2$ for each clause.
The set $E_2$ is the set of remaining edges.
Define $P$ as the path obtained by an Euler tour defined by a DFS of $G$.
Set $\Delta=1$ and place a breakpoint on each vertex of $P$.
Finally, set $k = \frac{|E_2|}{2}+3 (n_{SAT}-1)$.
This finalizes the construction.

\begin{theorem}
\label{thm:hardness}
Let $\curveP: [0,1] \rightarrow \RR^2$ be a polygonal curve of complexity $n$ with breakpoints $0 \leq t_1,\dots,t_m\leq 1$.
It is NP-complete to decide whether there exist a set $C$ of points in $\RR^2$ such that $\phi(P,C) \leq \Delta$ and $|C| \leq k$ for given $\Delta \in \RR$ and $k \in \mathbb{N}$.
\end{theorem}

\begin{figure}[h]
    \centering
    \includegraphics[width=.8\textwidth]{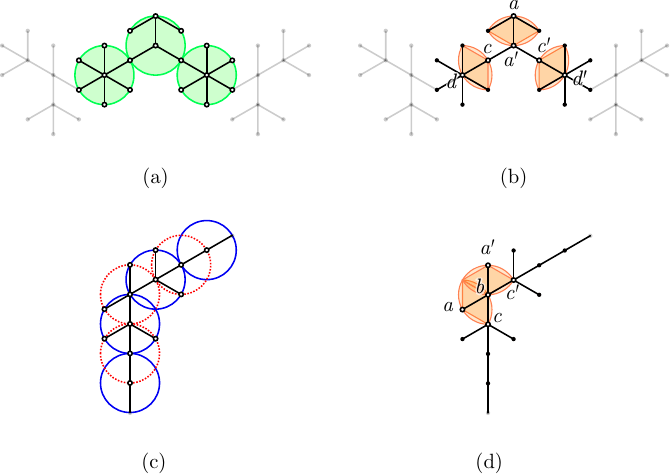}
    \caption{(a) and (b) show the separator gadget in black, and some edges of the adjacent variable gadgets in gray. (c) and (d) show the turn gadget.  Disks indicate potential optimal solutions.}
    \label{fig:separator}
\end{figure}

\begin{figure}[h!]
    \centering
    \includegraphics[width=0.95\textwidth]{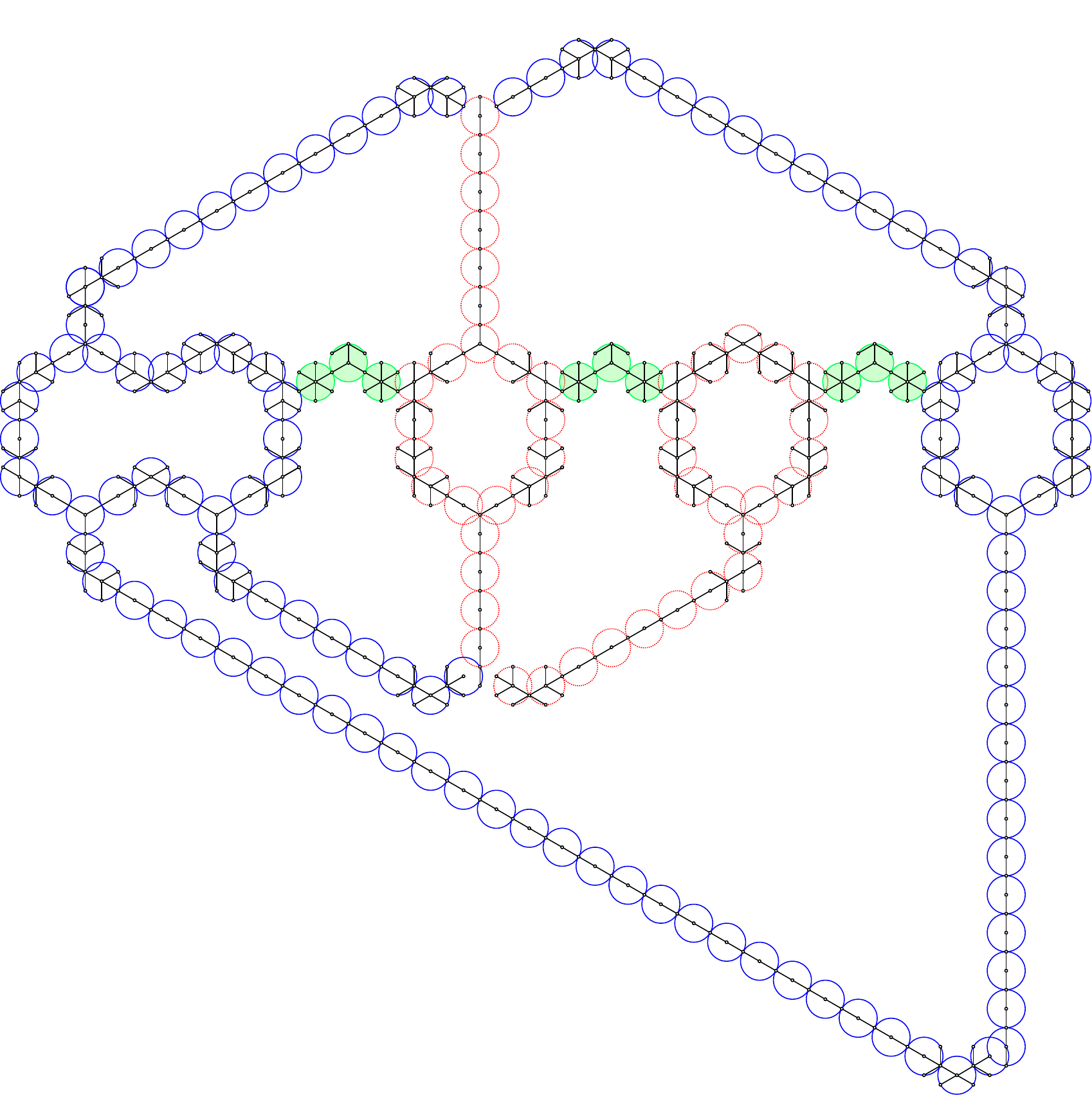}
    \caption{Example of reduction from $(x_1\vee x_2\vee x_4)\wedge(\overline{x_1}\vee \overline{x_2}\vee \overline{x_3})\wedge(\overline{x_1}\vee \overline{x_4})$.
    The centers of the disks are an optimal solution to the instance.}
    \label{fig:example}
\end{figure}

\begin{proof}
\smallskip
\noindent\textbf{{$(\Rightarrow)$}}
We assume that $A$ admits a positive solution, and constructs a positive solution $C$ for $B$ as follows.
For each variable $x_i$, add all the odd (even) points in the spine of the corresponding variable gadget to $C$ if $x_i$ is set to \texttt{true} (\texttt{false}) in $A$'s solution.
Do the same for all wire gadgets and the portion of the clause gadgets corresponding to $x_i$.
For each separator gadget, add the three green points shown in Figure~\ref{fig:separator} to $C$.
This finalizes the construction of $C$.
By construction, $|C|=k$, the variable, wire and separator gadgets are covered by disks centered at $C$. 
Because every positive (negative) clause in $A$ is adjacent to a variable assigned \texttt{true} (\texttt{false}), the clause segment is covered by a disk centered at the spine of a incident wire. 
Then, $C$ is a positive solution for $B$.

\smallskip
\noindent\textbf{$(\Leftarrow)$}
We assume that $B$ admits a positive solution $C$, and constructs a positive solution for $A$.
We first show we can construct a canonical solution $C'$ from $C$ with $|C'|\le|C|$.

Consider the separator gadget in Figure~\ref{fig:separator}~(b). 
The positions of the centers of unit disks that cover segment $aa'$ form a lune defined by the intersection of the unit disks centered at $a$ and $a'$.
The analogous is true for segments $cd$ and $c'd'$.
Such lunes are disjoint, hence $C$ has 3 distinct points, $c_1$, $c_2$ and $c_3$, to cover such segments.
Note that they cannot cover segments outside of the separator gadget.
We can move them to $a$, $d$ and $d'$ so that the set of segments that they cover is either the same or a superset of the previously covered segments.
The following assumes that (i) every separator gadget is covered by three points in $C$ as in Figure~\ref{fig:separator}~(a).

Consider the turn gadget in Figure~\ref{fig:separator}~(d).
Assume that $ab$ and $a'b$ are respectively covered by different points $c_1$ and $c_2$ in $C$.
Then, we can move $c_1$ and $c_2$ to $c$ and $c'$ while covering the same segments and possibly more.
Now, assume that $ab$ and $a'b$ are covered by $c_1\in C$.
Then, $c_1$ is in the intersection of the two lunes shown in Figure~\ref{fig:separator}~(d).
The only segments it can cover are $ab$, $a'b$, $cb$ and $c'b$. Then we can move it to $b$ without decreasing its coverage.
Assuming that turn gadgets are in canonical form, we can apply the moving argument at each remaining segment of $P$ in order to obtain a canonical solution $C'$, moving each point $c\in C$ to a breakpoint.

By (i), $C'$ has $\frac{|E_2|}{2}$ points to cover segments with endpoints at literal vertices. 
Note that each point can cover at most 2 edges in $E_2$.
Then each point must cover exactly 2 edges in $E_2$.
It follows that, for each variable in $A$, $C'$ contains points at either all positive literal vertices and none at negative vertices, or all negative literal vertices and none at positive vertices.
Because the clause segments are covered, the clauses of $A$ are satisfied, and we can obtain a solution for $A$.
That concludes the proof of NP-hardness.

The problem is in NP since verifying whether a given set $C$ is a solution for our problem can be done by computing the $\Delta$-free space diagrams for each curve in $C$ and $P$, and greedily partitioning $P$, verifying whether it is covered. This concludes the proof of the theorem.
\end{proof}

\section{Acknowledgements}
This research was initiated at the Eighth Annual Workshop on Geometry and Graphs, held at the Bellairs Research Institute in Barbados, January 31 – February 7, 2020. The authors are grateful to the organizers and to the participants of this workshop. Anne Driemel acknowledges funding from the DFG (project nr. 313421352).  Erin Chambers acknowledges funding from the National Science Foundation under grants AF-1907612, AF-2106672, and DBI-1759807.

\bibliographystyle{alpha}
\bibliography{bibliography}

\end{document}